\RequirePackage{fix-cm}

\documentclass[smallextended]{svjour3}       

\smartqed  

\usepackage[caption = false]{subfig}
\usepackage{graphicx,enumerate,appendix,amssymb,mathtools,tcolorbox,enumitem,float,lineno}
\usepackage[misc,geometry]{ifsym} 

\newcommand{\IR}{\mathbb{R}}

\usepackage[total={6in,9in}]{geometry}
\begin{document}



\title{Unveiling the dynamics of canard cycles and global behaviour in a singularly perturbed predator-prey system with Allee effect in predator}

\titlerunning{Canard cycles and global dynamics}        

\author{Tapan Saha \and Pallav Jyoti Pal}

\authorrunning{T Saha, P J Pal} 
\institute{
Tapan Saha \at
             Department of Mathematics, Presidency University, Kolkata -700073, India \\
              \email{tapan.maths@presiuniv.ac.in}           
           \and
Pallav Jyoti Pal (Corresponding author) \at
            Department of Mathematics, Krishna Chandra College, Hetampur-731124, Birbhum, India\\
              \email{pallav.pjp@gmail.com}           
                   }
\date{Received: date / Accepted: date}

\maketitle
				
\begin{abstract}
 In this article, we have considered a  planar slow-fast modified Leslie-Gower predator-prey model with a weak Allee effect in the predator, based on the natural assumption that the prey reproduces far more quickly than the predator. We present a thorough mathematical analysis demonstrating the existence of homoclinic orbits, heteroclinic orbits,  singular Hopf bifurcation, canard limit cycles, relaxation oscillations, the birth of canard explosion  by combining the normal form theory of slow-fast systems,  Fenichel’s theorem and blow-up technique near non-hyperbolic point. We have obtained very rich dynamical phenomena of the model, including the saddle-node, Hopf, transcritical bifurcation, generalized Hopf, cusp point, homoclinic orbit, heteroclinic orbit, and Bogdanov-Takens bifurcations.
 Moreover, we have investigated the global stability of the unique positive equilibrium, as well as bistability, which shows that the system can display either ``prey extinction", ``stable coexistence", or ``oscillating coexistence" depending on the initial population size and values of the system parameters. The theoretical findings are verified by numerical simulations.
\end{abstract}
\section{Introduction}\label{Slowfast_intro} 
Singularly perturbed systems of ordinary differential equations may be used to predict the evolution of a wide range of physical and applied systems with multiple timescales. Such a system can be written in the standard form as follows:
\begin{subequations}\label{general_slow_fast}
\begin{align}
	\dot{x}& = F\left(x, y,\mu, \epsilon\right),\\
\dot{y} &= \epsilon G\left(x, y, \mu, \epsilon\right),
 \end{align}
\end{subequations}
where $(x,y) \in \mathbb{R}^m\times \mathbb{R}^n$ such that $x, y$ are the fast and the slow variables respectively, $\mu\in\IR^k$ are system parameters, $m,n,k\geq 1$,  $F$ and $G$ are the sufficiently smooth functions,  $0 < \epsilon \ll 1$ is the singular perturbation parameter, and the over dot $(~ \dot{}~ )$ stands for (fast) time derivative $\frac{d}{dt}$. A powerful mathematical framework for studying slow-fast systems \eqref{general_slow_fast} is known as Geometric Singular Perturbation Theory (GSPT). GSPT encompasses a wide variety of geometric methods for doing so, namely, Fenichel theory \cite{fenichel1979geometric},  blow-up method \cite{dumortier1996canard,krupa2001extending,krupa2001relaxation}, slow-fast normal form theory \cite{arnold1994dynamical}. 
For $\epsilon\to 0$, the limiting subsystem obtained from \eqref{general_slow_fast} is a fast subsystem (or layered system) $\dot{x} = F\left(x, y,\mu, 0\right)$ where the slow variables $y$ acting as parameters. By rescaling time from $t$ to $\tau=t/\epsilon$, the fast to the slow timescale in \eqref{general_slow_fast}, an equivalent system to \eqref{general_slow_fast} is obtained which yields a differential-algebraic equation (called the slow subsystem associated with \eqref{general_slow_fast}) for the singular limit $\epsilon\to 0$.  The slow subsystem is a dynamical system on the set $  M_0 = \{(x, y) \in \IR^m \times \IR^n : F(x, y, \mu, 0) = 0\}$. This is also the set of equilibria of the fast subsystem, with $y$ acting as a parameter. We refer to $M_0$ as a critical manifold if it is a submanifold. Normal hyperbolicity is a crucial property that manifold $M_0$ may have. A point $p\in M_0$ is an equilibrium point of the fast subsystem. If all the eigenvalues of the $m\times m$ matrix $(D_xF)(p)$ have non-zero real parts, then we say that $M_0$ is normally hyperbolic at the point $p\in M_0$. When all the eigenvalues of the $m\times m$ matrix $(D_xF)(p,\mu, 0)$ have negative real parts for $p\in S\subset M$, then we say that $S\subset M$ is attracting, and when all the eigenvalues have positive real parts, then we say that $S$ is repelling. When $M_0$ is a normally hyperbolic critical manifold, Fenichel’s theorems is applied as a regular perturbation corresponding to the singular system near $M_0$, and it says that $M_0$ is perturbed to the invariant slow manifolds $M_\epsilon$ which is at a distance $\mathcal{O}(\epsilon)$  away from $M_0$. What this means is that as $\epsilon$ approaches zero, the flow on the (locally) invariant manifold $M_\epsilon$ converges to the slow subsystem on the critical manifold $M_0$. 
\begin{figure}[H]
\setlength{\belowcaptionskip}{-10pt}
\centering
	\subfloat[]{\includegraphics[width=8cm,height=7cm]{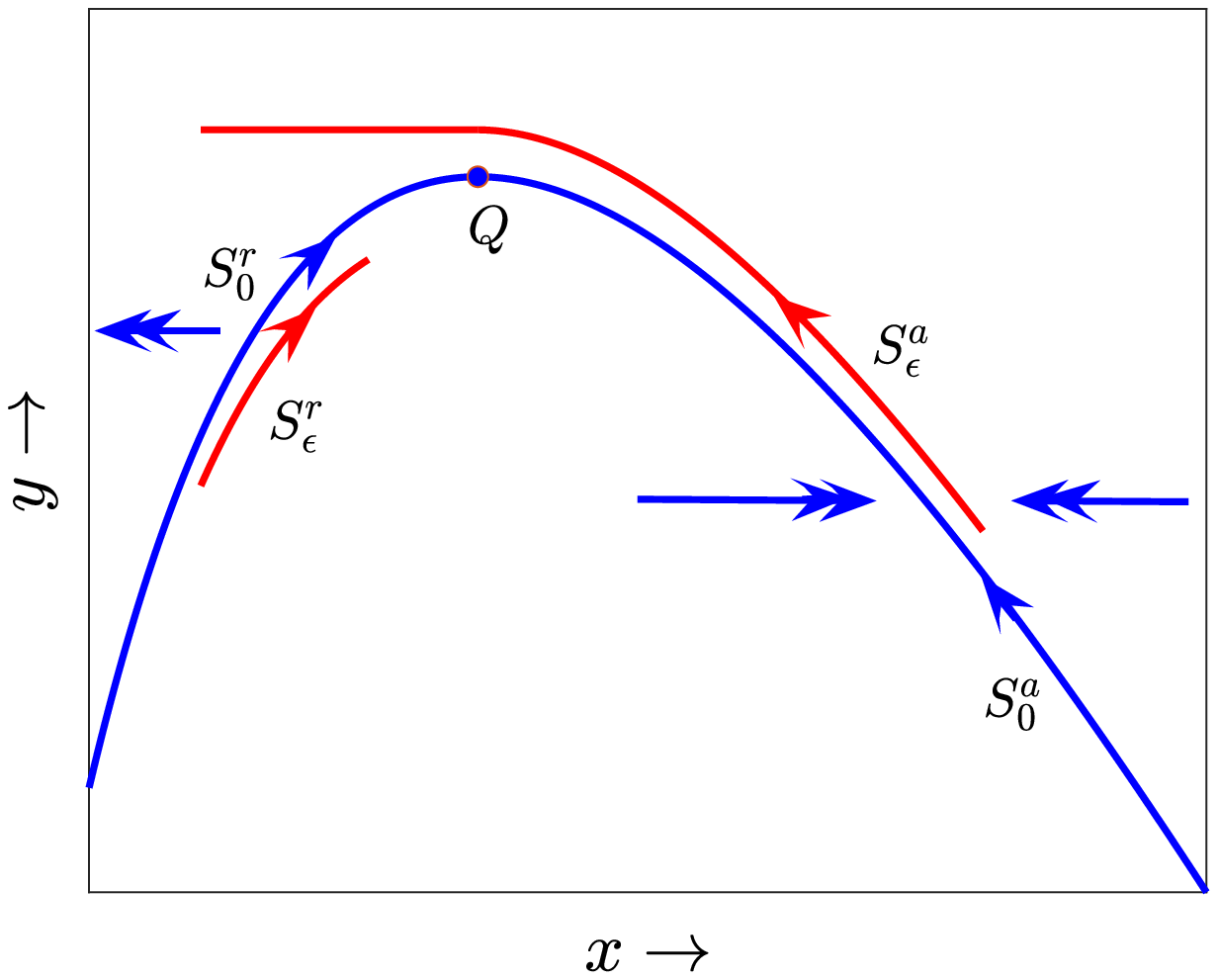}}
		\subfloat[]{\includegraphics[width=8cm,height=7cm]{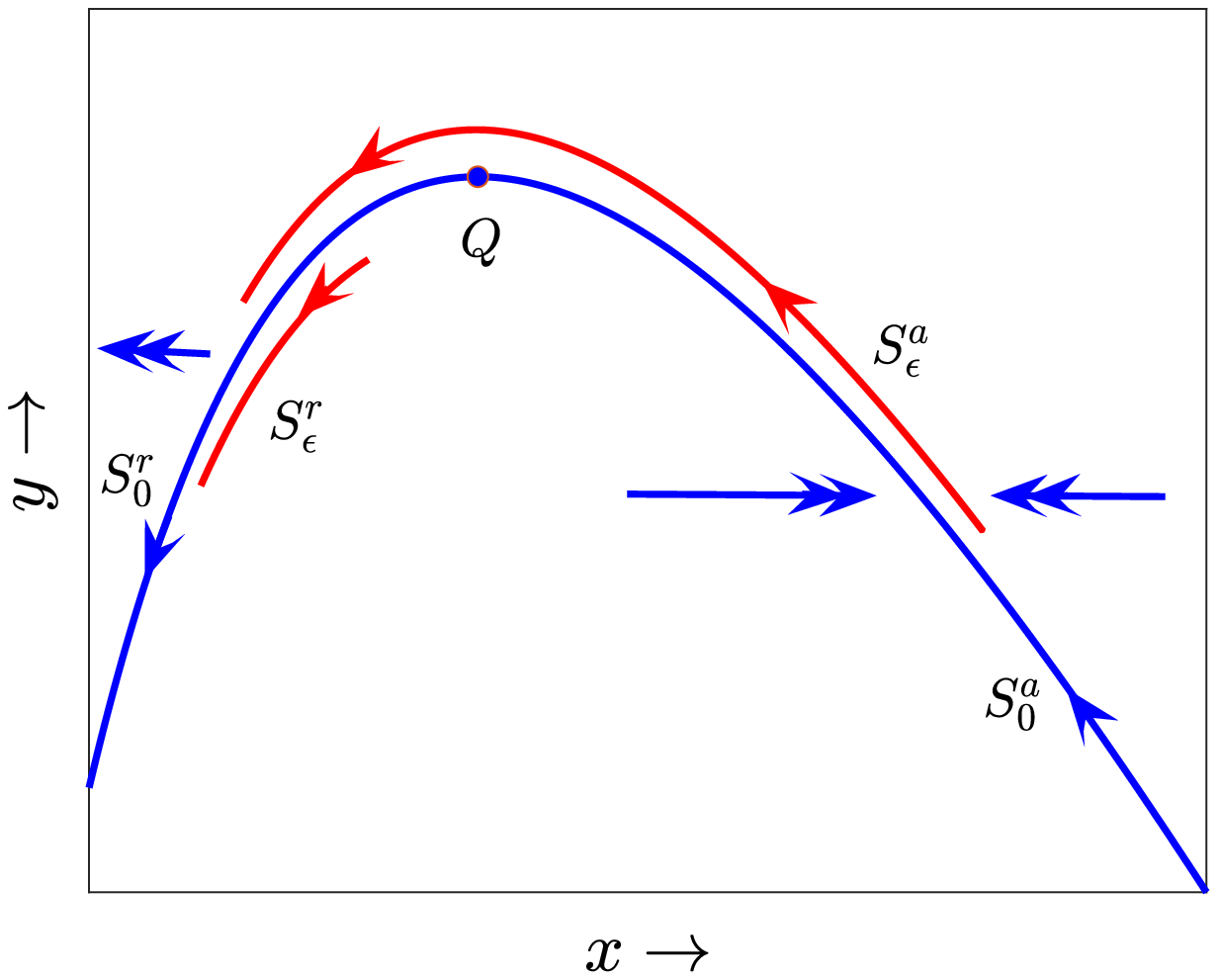}}
\caption{The parabolic critical manifold $M_0=S_0^a\cup N\cup S_0^r$ (shown by blue curve) where $S_0^a$ and $S_0^r$ are the attracting and repelling submanifolds, respectively, and    the normally non-hyperbolic point is $Q (x_m,y_m)\in N $ (shown by a blue dot). The double arrows represent fast flow, and the single arrows represent slow flow. (a) The slow manifolds $S_\epsilon^a$ and $S_\epsilon^r$  near the jump point $Q$  are represented by the red curve. (b) The slow manifolds $S_\epsilon^a$ and $S_\epsilon^r$ (shown by red curve) near the canard point $Q$. (For interpretation of the references to colour in this figure caption, the reader is referred to the web version of this chapter.)} 
\label{fig:jump_canard}			
\end{figure}

In the non-normally hyperbolic domain, however, Fenichel-Tikhonov theory fails. Suppose, $M_0=S_0^a\cup N\cup S_0^r$, where $S_0^a$ and $S_0^r$ are the attracting and repelling branch of $M_0$ and $N$ is a non-normally hyperbolic point or submanifold. Away from the non-normally hyperbolic singularity, Fenichel's theorem shows that, for $0<\epsilon\ll 1$, $S_0^a$ and $S_0^r$ are smoothly perturbed to invariant manifolds $S_\epsilon^a$ and $S_\epsilon^r$ respectively. For instance, in the most typical scenario, the non-hyperbolic singularities $Q\in N$ are  jump point, canard point \cite{krupa2001extending,kuehn2015multiple}, etc. In the case $Q\in N$, the blow-up method, introduced by Dumortier and Roussarie \cite{dumortier1996canard} and developed by Krupa and Szmolyan \cite{krupa2001extending,krupa2001relaxation} is commonly used to investigate the dynamics of the slow-fast system  \eqref{general_slow_fast} where the non-hyperbolic singularities $Q\in N$ are de-singularized by using this blow-up method.
Such desigularization enables one to explore the dynamics in the non-normally hyperbolic domain using classical approaches such as regular perturbation and centre manifold theory for the study of dynamical systems. The scenario of loss of normal hyperbolicity is effectively significant as it is associated to dynamic properties like relaxation oscillations, canards, heteroclinic orbits, homoclinic orbits, etc \cite{kuehn2015multiple,krupa2001extending,krupa2001relaxation,zhao2022relaxation,kuehn2010from,atabaigi2021canard,saha2021relaxation}. 
 A generic fold point $Q\in M_0$ is referred to as a jump point if a candidate orbit follows first the attracting branch $S_0^a$ closely, reaches the vicinity of the fold point $Q$, and then follows the direction of the fast flow abruptly away from $Q$, as shown in Fig.\ref{fig:jump_canard}a. 
A canard point is a fold point $Q$ for which $G(Q,\mu,0)=0$ for some $\mu$. In this case, it could happen that the attracting slow manifold $S_\epsilon^a$ will remain in close to the repelling slow manifold $S_\epsilon^r$ for a time of $\mathcal{O}(1)$ (see Fig. \ref{fig:jump_canard}b). Such solutions are known as canards. There is a possibility that there are certain values of $\mu(\epsilon)$ for which the attracting slow manifold $S_\epsilon^a$ connects to the repelling slow manifold $S_\epsilon^r$. The term \lq\lq maximal canard\rq\rq is used to describe such solutions \cite{krupa2001relaxation}. Another well-known occurrence in this setting is relaxation oscillations, when solutions approach to a fold point slowly but then abruptly jump from the fold point to another stable branch of $M_0$, then follow the slow dynamics again until a new fold point is reached, and so on, and finally forming periodic orbits \cite{krupa2001relaxation,zhao2022relaxation,kuehn2015multiple}. A quick shift upon change of a control parameter from a small amplitude limit cycle via canard cycles to a large amplitude relaxation oscillation may occur for the system \eqref{general_slow_fast} within an exponentially narrow range $\mathcal{O}(e^{-1/\epsilon})$ of the control parameter. It is referred to as Canard explosion.

In this article, under the natural assumption that the prey reproduces considerably quicker than the predator, the primary emphasis is on planar slow-fast predator-prey systems with two time scales of the type \eqref{general_slow_fast} where $m=n=1$. A significant amount of research has been put into investigating the canard phenomena and the existence of relaxation oscillations of planar slow-fast predator-prey systems. The followings are just a few instances, by no means exhaustive, where this kind of investigation has been done. Hek \cite{hek2010geometric} applied the Fenichel's theory to biology. Using asymptotic expansion techniques, Kooi and Poggiale \cite{kooi2018modelling} demonstrated how to locate a canard solution at the turning point in the Rosenzweig-MacArthur model on two time scales.  
In Ambrosio et. al \cite{ambrosio2018canard}, authors considered a slow-fast predator-prey model of modified Leslie–Gower type with two time scales. By using the blow-up method, they are able to clearly display the behaviour close to the fold point and  demonstrated that the limit-cycle experiences the canard phenomena while crossing the folded node. The dynamics of a slow-fast predator-prey model are investigated in \cite{atabaigi2021canard}, where the predator is a generalist predator that feeds on both the focal prey and the functional response is Holling type III. Using tools like the theory of normal forms for slow-fast systems, the theory of geometric singular perturbations, and the blow-up method,
the author explores the existence of relaxation oscillations and canard limit cycles bifurcating from singular homoclinic cycles.

The Allee effect has been the subject of several publications on predator-prey system \cite{courchamp2008allee,terry2015predator,zhou2005stability,rahmi2021modified,feng2015dynamics,hadjiavgousti2008allee,pal2015qualitative,gonzalez2011multiple,sardar2022allee}. Most studies among them have considered the Allee impact of the prey population growth. Many observations, however, suggest that the Allee effect is also evident in the population of predators, for instance, Seabirds and the African wild dog (Lycaon pictus) \cite{courchamp2008allee}. There has been little research on the impact of the Allee effect on predator populations \cite{terry2015predator,zhou2005stability,rahmi2021modified,feng2015dynamics}. To the best of our knowledge, there is no literature on slow-fast predator-prey model where predator population growth is affected by weak Allee effect. By the term $\Phi(v)=\frac{v}{v+m}$, often known as the weak Allee effect function with $m$ as the Allee effect constant, we introduce an Allee effect into the predator equation. $\Phi(v)$ measures the probability that a female predator will come into contact with at least one male and mate with him during the reproductive stage. 
This Allee effect function reduces the predator's per capita growth rate from $s$ to $\frac{sv}{v+m}$. The Beddington–DeAngelis functional response $\Psi(u,v)=\frac{mv}{a+bu+cv}$ is comparable to the well-known Holling type II functional response $\Psi(u,v)=\frac{mv}{a+bu}$, with the addition of an additional factor $cv$ in the denominator. Here, $u = u(t)$ and $v = v(t)$ respectively denote the prey and predator population densities, $m$ denotes the maximum per capita consumption rate of a predator, both $a$ and $b$ are prey saturation constants, $c$ is the predator interference. The factor $cv$ reflects the mutual interference between predators.
The Beddington-DeAngelis functional response also avoids the controversial problem that the ratio-dependent functional response $\Psi(u,v)=\frac{mv}{bu+cv}$ have at low population densities. 

We then arrive at the following modified Leslie--Gower predator--prey model with logistic growth for both the prey and Allee effect in predator given by:
\begin{subequations}\label{sf_model1}
	\begin{align}
		\frac{du}{dT}&=ru\left(1- \frac{u}{K}\right) -\frac{muv}{a + bu +cv},\\
		\frac{dv}{dT}&=sv\left(\frac{v}{n+v}-\frac{v}{d+hu}\right),
	\end{align}
\end{subequations}
subjected to initial conditions $u(0)\geq 0$, $v(0)\geq 0$, parameters $(r, K, m, n, a, b, c, s, d, h)\in \IR_+^{10}$ such that $u = u(t)$ and $v = v(t)$ respectively denote the prey and predator population densities at time $t > 0$. The Allee effect is considered in predator population because the predator population is more prone than their prey \cite{terry2015predator}. 
Here, $r$ is the intrinsic per capita growth
rate of prey,
$K$ is the environmental carrying capacity,
  $h$ measures of the food quality,
$d$ is the amount of alternative food available for predators, and the meaning of other parameters are already mentioned above.  

Non-dimensionalizing the system \eqref{sf_model1} by using the following rescaling transformations:
	\begin{align}
		t = rT,~~ x = \frac{u}{K}, ~~y = \frac{cv}{bK},
	\end{align}
we have
\begin{subequations}\label{sf_model2}
	\begin{align}
		\frac{dx}{dt}&=x\left(1-x\right)-\frac{\alpha xy}{\beta+x+ y}=f(x,y,\mu),\\
		\frac{dy}{dt}&=\epsilon y\left(\frac{y}{y+\gamma}-\frac{y}{\delta + \theta x}\right)=\epsilon g(x,y,\mu),
	\end{align}
	\end{subequations}
where $x, y$ are the new dimensionless variables, $\mu=(\alpha, \beta, \gamma, \delta, \theta)$ with $\alpha=\frac{m}{rc},\beta=\frac{a}{bk}$, $\gamma=\frac{nc}{bK}$, $\delta=\frac{cd}{bK}$, $\theta=\frac{ch}{b}$ and $\epsilon=\frac{s}{r}$. The parameters are positive with $0<\epsilon\ll 1$.

The remaining part of this chapter is organized as follows:
In Sect. \ref{Sec:Preliminaries}, we present some basics results  for the system \eqref{sf_model2}.
The slow-fast system is analysed in Sect. \ref{Sec:sf_analysis}. The existence of the singular Hopf bifurcation and canard cycles are investigated in Sect. \ref{Sec.hopf_canard}.
In Sect. \ref{Sec:hetro_homo_orbits}, we also provide thorough proof of the existence of heteroclinic and homoclinic orbits. In Sect \ref{Sec:relaxation}, we  prove the existence of relaxation oscillation and the bistability phenomenon. The main theoretical predictions are verified using numerical simulations in appropriate sections.  Finally, some brief conclusions of our findings are presented in Sect. \ref{Sec:conclusion}.

\section{Basic Results} \label{Sec:Preliminaries}
In this section, we discuss some basic results, including the invariance, boundedness, existence of equilibria and their nature, and bifurcation scenario for the system \eqref{sf_model2}.  
\begin{lemma}\label{lemma_1}
The first quadrant $\IR^2_+=\{(x,y)\in\IR^2|x\geq 0, y\geq 0\}$ is invariant under the flow generated by the vector field 
$V_{\epsilon, \mu}=f\frac{\partial}{\partial x}+\epsilon g\frac{\partial}{\partial y}$.
\end{lemma}

\begin{lemma}\label{lemma_2}
All the solutions of the model system \eqref{sf_model2} initiated from the interior of $\IR^2_+$ are bounded.
\end{lemma}

The system \eqref{sf_model2} has three equilibria on the co-ordinate axes, namely, the trivial equilibrium $E_0(0,0)$ and the boundary equilibria $E_1(1,0)$, $E_{2}(0,\delta-\gamma)$ where $E_2$ exists if $\delta>\gamma$. We have the following trivial results on the nature of the equilibria on the co-ordinate axes.

\begin{lemma}\label{lemma_5}
\begin{enumerate}[label=(\roman*)]
\item The trivial equilibrium $E_0(0,0)$ is a saddle node.

\item The boundary equilibrium $E_{1b}(1,0)$ is a saddle node. 

\item The boundary equilibrium $E_{2b}(0, \delta-\gamma)$ is a hyperbolic stable node if $\delta>\gamma+\frac{\beta}{\alpha-1}$, a hyperbolic saddle if $\delta<\gamma+\frac{\beta}{\alpha-1}$ and a saddle node if 
$\delta=\gamma+\frac{\beta}{\alpha-1}$ 
\end{enumerate}
\end{lemma}

The interior equilibria are the points of intersection of the non-trivial prey and predator nullclines (see Fig. \ref{fig:nullcline}) given by 
\begin{align*}
y&=\frac{(\beta+x)(1-x)}{\alpha+x-1},\\
y&=\theta x+\delta-\gamma.
\end{align*}



Assuming
\begin{align*}
D&=(\alpha\theta+\delta-\gamma+\beta-\theta-1)^2-4(1+\theta)\left(\alpha(\delta-\gamma)-(\delta-\gamma)-\beta\right),
\end{align*}
we consider the following parametric regions 
\begin{subequations}\label{parametric conditions}
\begin{align}
    R_1&=\left\{\mu\middle|D>0,\, 
    \alpha>\frac{1}{1-\beta}, \beta<1, \delta>\gamma+\frac{\beta}{\alpha-1},\theta(\alpha-1)+\frac{\alpha\beta}{\alpha-1}<1\right\},\\
    R_2&=\left\{\mu\middle|D=0,\,\alpha>\frac{1}{1-\beta}, \beta<1, \delta>\gamma+\frac{\beta}{\alpha-1}\right\},\\
    R_3&=\left\{\mu\middle|\alpha>\frac{1}{1-\beta}, \beta<1, \{0\leq\gamma-\delta<1\}\cup\{\gamma<\delta<\gamma+\frac{\beta}{\alpha-1}\}\right\},\\
    R_4&=\left\{\mu\middle|\{D<0\}\cup\{\gamma\geq\delta+1\}\cup\{D>0,\alpha>\frac{1}{1-\beta},\beta<1, \delta>\gamma+\frac{\beta}{\alpha-1},\right.\nonumber\\
    & \left. (\delta-\gamma)+\theta(\alpha-1)>1-\beta\} \right\}.
\end{align}
\end{subequations}

We now state the following results on the existence and stability of the interior equilibria of the system \eqref{sf_model2}.

\begin{lemma}\label{lemma_6}
\begin{enumerate}[label=(\roman*)]
\item If $\mu\in R_1$ then there exist two interior equilibrium points $E_{1*}(x_{1*}, y_{1*})$ and $E_{2*}(x_{2*}, y_{2*})$, where 
\begin{align*}
x_{1*}&= \frac{\theta+1-\alpha\theta-\beta+\gamma-\delta-\sqrt{D}}{2(1+\theta)},\,\, y_{1*}=\delta-\gamma+\theta x_{1*},\\
x_{2*}&=\frac{\theta+1-\alpha\theta-\beta+\gamma-\delta+\sqrt{D}}{2(1+\theta)},\,\, y_{2*}=\delta-\gamma+\theta x_{2*}.
\end{align*}
The equilibrium $E_{1*}$ is a hyperbolic saddle and $E_{2*}$ is a stable equilibrium point if $x_{2*}\geq x_m$. For $x_{2*}<x_m$, the equilibrium $E_{2*}$ will be either stable or unstable, depending on whether ${\text Trace} J|_{E_{2*}}< {\text or} >0$.

\item If $\mu\in R_2$ then there exists only one interior equilibrium point $\bar{E}(\bar{x}, \bar{y})$, where
\begin{align*}
\bar{x}&=\frac{\theta+1-\alpha\theta-\beta+\gamma-\delta}{2(1+\theta)},\,\,
\bar{y}=\delta-\gamma+\theta \bar{x}.
\end{align*}
In this case, the non-trivial predator nullcline touches the non-trivial prey nullcline tangentially at the point $\bar{E}$. The equilibrium $\bar{E}$ is a saddle node.

\item If $\mu\in R_3$ then there exists only one interior equilibrium point $E_{*}(x_{*}, y_{*})$, where
\begin{align*}
x_{*}&=\frac{\theta+1-\alpha\theta-\beta+\gamma-\delta+\sqrt{D}}{2(1+\theta)},\,\, y_{*}=\delta-\gamma+\theta x_{*}.
\end{align*} The equilibrium $E_*$ is stable if $x_*\geq x_m$ and for $x_*<x_m$, it will be either stable or unstable depending on whether ${\text Trace} J|_{E_{*}}< {\text or} >0$.

\item  If $\mu\in R_4$ then the system \eqref{sf_model2} has no interior equilibrium in $\IR^2_+$.
\end{enumerate}
\end{lemma}

\begin{figure}[H]
\setlength{\belowcaptionskip}{-10pt}
\centering
			\includegraphics[width=10cm,height=8cm]{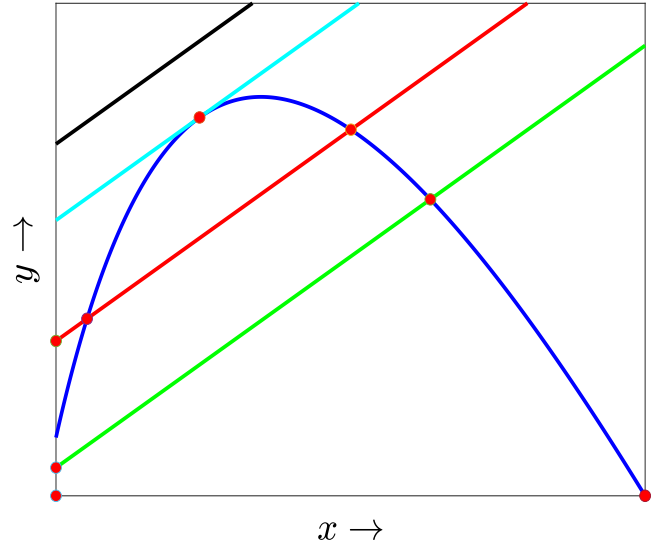}
\caption{In this representation, relative position of the non-trivial nullclines are shown as solid lines and for clarity, we have not included trivial nullclines. The $x$ and $y$ axes are the densities of prey and predator species, respectively. Non-trivial prey nullcline is shown by the blue curve, whereas non-trivial predator nullclines are shown by  straight lines for variable $\delta$. The coexistence equilibria are shown by solid red circles. For the given parameter values of $\alpha=1.5, \beta=0.0207, \gamma=0.3,  \theta=0.3$ and variable $\delta$, the figure shows that the number of interior equilibrium points ranges from $0$ to $2$. Different relative positions of predator nullclines are shown in different colours for various values of $\delta$: black for $\delta=0.55$ ($\mu\in R_4$, no interior equilibrium), cyan for $\delta=0.49576955$ ($\mu\in R_2$, unique interior equilibrium), red for $\delta=0.41$ ($\mu \in R_1$, two interior equilibrium), green for $\delta=0.32$ ($\mu\in R_3$, unique interior equilibrium). (For interpretation of the references to colour in this figure caption, the reader is referred to the web version of this chapter.)}
\label{fig:nullcline}			
\end{figure}

\subsection{Bifurcation Scenario}

The non-trivial prey and predator nullclines intersect the positive $y$-axis at the point $P(0, \frac{\beta}{\alpha-1})$ and $E_2(0, \delta-\gamma)$ and consequently, based on the nature of the non-trivial nullclines we have that if $E_{2b}$ lies below the point $P$ then there always exists a unique interior equilibrium point $E_{*}$, if $E_{2b}$ lies above the point $P$ then under certain parametric conditions (as mentioned in {\bf Lemma} \ref{lemma_6}) there may exist zero, one or two interior equilibrium points. Thus, we see that varying the control parameter $\delta$ it follows that for $\delta=\delta_{TC}=\gamma+\frac{\beta}{\alpha-1}$, the model system \eqref{sf_model2} undergoes a transcritical bifurcation as one interior equilibrium bifurcates from $E_2(0,\delta-\gamma)$ as $\delta$ passes through $\delta=\delta_{TC}$. Assuming the parametric conditions $\alpha>\frac{1}{1-\beta}, \beta<1, \delta>\gamma+\frac{\beta}{\alpha-1},\theta(\alpha-1)+\frac{\alpha\beta}{\alpha-1}<1$, we have that for $D>0$, there exist two interior equilibrium points $E_{1*}$ and $E_{2*}$ where $E_{1*}$ is a hyperbolic saddle point; for $D=0 (\theta=\theta_{SN})$, the two equilibrium points $E_{1*}, E_{2*}$ coalesce at the degenerated saddle node equilibrium point $\bar{E}(\bar{x},\bar{y})$ and for $D<0$ there exists no equilibrium point. Thus, we have saddle node bifurcation of equilibria, i.e., the model system \eqref {sf_model2} undergoes a saddle node bifurcation as $\theta$ passes through $\theta=\theta_{SN}$. For $(\delta, \theta)=(\delta_{TC}, \theta_{SN})$, the model system \eqref {sf_model2} undergoes a saddle-node-transcritical bifurcation topologically equivalent to co-dimension 2 cusp bifurcation as $(\delta, \theta)$ passes through $(\delta, \theta)=(\delta_{TC}, \theta_{SN})$. Now, it may also happen that varying $\delta$, there may take place Hopf bifurcation around $E_{2*}$ (or $E_*$) for $\delta=\delta_H$ and will be studied in the next section in the realm of slow-fast analysis. We also have that for $D=0$, Trace \,$J(\bar{E})=0 \left((\delta, \theta)=(\delta_{H}, \theta_{SN})\right)$, the equilibrium $\bar{E}$ is a Bogdanov--Takens (BT) singularity and thus, varying the parameter $(\delta, \theta)$ in a neighbourhood of $(\delta, \theta)=(\delta_{H}, \theta_{SN})$, various codimension-2 BT bifurcation phenomena (emergence and destruction of periodic orbit, homoclinic orbit) will be observed. Following \cite{kuznetsov1998elements} one can explicitly compute the normal forms of the various bifurcations mentioned here and verify the results analytically. But, as the chapter aims to investigate the dynamics of a slow-fast system in the realm of GSPT and blow-up technique, we present below the two-parameter bifurcation diagram for the various bifurcation results. 

\begin{figure}[H]
\setlength{\belowcaptionskip}{-10pt}
\centering
\includegraphics[width=15cm,height=10cm]{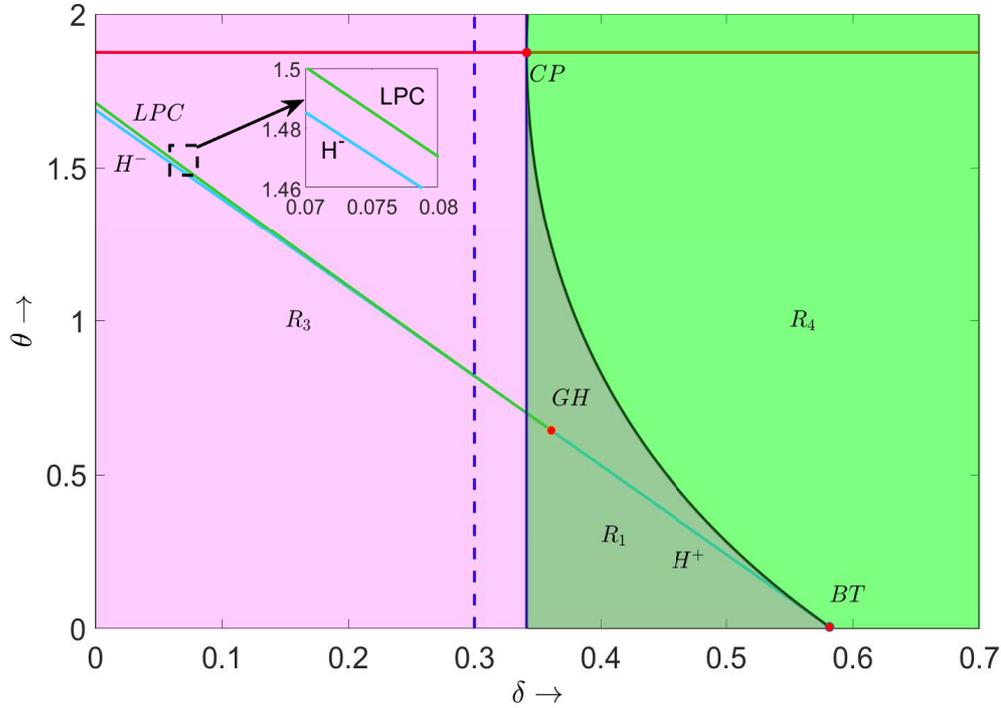}
\caption{Two-parameter Bifurcation diagram in $\delta-\theta$ parameter plane. The Hopf (H) bifurcation curve (cyan) intersects at the Generalized Hopf (GH) bifurcation point located at $(\delta_{GH}, \theta_{GH})=(0.361212, 0.638870)$ in the region $R_1$ with the limit point of cycles (LPC) bifurcation curve (green). The thick black curve represents the saddle-node (SN) bifurcation curve, and the thick blue line is the transcritical bifurcation curve (TC). The TC and SN curve intersect tangentially at a cusp point (CP). The broken blue line is the equation $\delta=\gamma$ determines the existence of boundary equilibrium point $E_{2b}$.  The horizontal red line represents the equation $\theta(\alpha-1)+\frac{\alpha\beta}{\alpha-1}=1$. The SN and H curves approach each other and eventually collide at a Bogdanov-Takens (BT) point located at $(\delta_{BT}, \theta_{BT})=(0.581662, 0.005247)$ when $\delta$ increases. The areas $R_3$, $R_1$, and $R_4$ correspond to the pink, olive green, and green regions, respectively, whereas the region $R_2$ is on the black SN curve. There also exists a Homoclinic curve originating from BT point, but not shown here as its range of existence is very narrow. The other parameter values are $\alpha=1.5$, $\beta=0.0207$ $\gamma=0.3$ and $\epsilon=0.01$. (For the interpretation of the colour references in this figure caption, the reader is referred to the web version of this chapter.)}
\label{fig:regions_bifurcation}			
\end{figure} 

\section{Slow-Fast Analysis}\label{Sec:sf_analysis}
With the time scaling $\tau = \epsilon t$, $0<\epsilon\ll 1$ the system \eqref{sf_model2} transforms to the following topologically equivalent system:
\begin{subequations}\label{sf_model3}
\begin{align}
	\epsilon\frac{dx}{d\tau}&=x(1-x)-\frac{\alpha xy}{\beta+x+y},\\
	\frac{dy}{d\tau}&=y\left(\frac{y}{y+\gamma}-\frac{y}{\delta+\theta x}\right).
\end{align}
\end{subequations}
The model system \eqref{sf_model2} or \eqref{sf_model3} is a standard form of slow--fast system with $t$ as the fast timescale and $\tau$ as the slow timescale, respectively. The variables $x$ and $y$ are referred as fast and slow variables, respectively. In the singular limit $\epsilon\rightarrow 0$, the systems \eqref{sf_model2} and \eqref{sf_model3} transform to the following fast and slow subsystems.
\begin{subequations}\label{sf_fast subsystem}
\begin{align}
	\frac{dx}{dt}&=x(1-x)-\frac{\alpha xy}{\beta+x+y},\\
	\frac{dy}{dt}&=0,
\end{align}
\end{subequations}
and
\begin{subequations}\label{sf_slow subsystem}
\begin{align}
	0&=x(1-x)-\frac{\alpha xy}{\beta+x+y},\\
	\frac{dy}{d\tau}&=y\left(\frac{y}{y+\gamma}-\frac{y}{\delta+\theta x}\right).
\end{align}
\end{subequations}
 
The slow flow corresponding to the slow subsystem \eqref{sf_slow subsystem} is constrained on the critical set $M_0$ given by
\begin{align*}
 M_0&=\left\{\left(x,y\right)\in \IR^2_+ \,\middle|\, f(x,y)=0\right\}.
 \end{align*}
The critical set $M_0$ consists of two kinds of critical manifolds given by
\begin{subequations}\label{sf_manifolds}
\begin{align}
	M_{10}&=\left\{\left(x,y\right)\in \IR^2_+ \,\middle|\, x=0\right\},\\
	M_{20}&=\left\{\left(x,y\right)\in \IR^2_+ \,\middle|\, y =\phi(x)=
	\frac{ (1-x)(\beta + x)}{\alpha+x-1}\equiv \phi(x), \alpha>\frac{1}{1-\beta}, \beta<1\right\}. \label{sf_manifold_m20}
\end{align}
\end{subequations}

We now have the following basic result on the nature of the function $\phi(x)$: 
\begin{lemma}\label{lemma_3}
 \begin{enumerate}[label=(\roman*)]
\item The function $\phi(x)$ decreases strictly in $\IR^2_+$ if $1<\alpha\leq \beta+1$.\\
\item The function $\phi_(x)$ has a local maxima at $x_m=1-\alpha +\sqrt{\alpha\left(\alpha-1-\beta\right)}$ in $\IR^2_+$ if $\alpha>\frac{1}{1-\beta}$, $\beta<1$. 
\end{enumerate} 
\end{lemma}

Henceforth, we will be assuming throughout the article the parametric condition 
that $\alpha>\frac{1}{1-\beta},\,\beta<1$, to ensure that the critical manifold $M_{20}$ is of  parabolic shape, increases in $0<x<x_m$ and decreases in $x_m<x<1$. The critical manifold $M_{20}$ looses its normal hyperbolicity at $P\left(0, \frac{\beta}{\alpha-1}\right)$ and $Q(x_m, y_m)$ (maximum point), $y_m=\phi(x_m)$. Consequently, it consists of two branches $S_0^r$ and $S_0^a$ where $S_0^r$ is the branch from $P$ to $Q$ and is hyperbolic repelling; $S_0^a$ is the branch from $Q$ to $R(1,0)$, and is hyperbolic attracting. Thus, 
\begin{align}
	S_0^r&=M_{20}\cap \left\{(x,y)\in \IR^2_+ \,\middle|\, 0<x<x_m\right\},\\
	S_0^a&=M_{20} \cap \left\{(x,y)\in \IR^2_+ \,\middle|\, x_m<x< 1\right\}.
\end{align}


Similarly, the normally hyperbolic repelling and attracting parts of the critical manifold $M_{10}$, denoted by $S_0^{r+}$ and $S_0^{a+}$ are given by
\begin{align}
    S_0^{r+}&=M_{10}\cap\left\{(x,y)\in \IR^2_+ \,\middle|\, 0<y<\frac{\beta}{\alpha-1}\right\},\\
	S_0^{a+}&=M_{10} \cap \left\{(x,y)\in \IR^2_+ \,\middle|\, y>\frac{\beta}{\alpha-1}\right\}.
\end{align}

The slow flow that evolves on the critical manifold $M_{20}$ is given by,
\begin{align}\label{sf_slow_sub}
		\frac{dx}{d\tau}&=\frac{\phi^2(x)\left((1+\theta)x^2+(\alpha\theta+\delta-\gamma+\beta-\theta-1)x+\alpha(\delta-\gamma)-(\delta-\gamma)-\beta\right)}{\phi'(x)(\delta+\theta x)(-x^2+x(1+\gamma-\beta)+(\alpha\gamma+\beta-\gamma)},
		\end{align}

and is not defined at the point $Q$. The point $Q$ is known as the fold point, because it corresponds to a fold bifurcation for \eqref{sf_fast subsystem} considering $y$ as a parameter. Now, for $0<\epsilon \ll  1$, Fenichel's theorem tells us that, $S_0^{r}$ and $S_0^a$ can be perturbed to $S_{\epsilon}^{r}$ and $S_{\epsilon}^a$ which are within $\mathcal{O}(\epsilon)$ distance from $S_0^{r}$ and $S_0^a$.
\begin{figure}[H]
\setlength{\belowcaptionskip}{-10pt}
\centering
			\includegraphics[width=14cm,height=8cm]{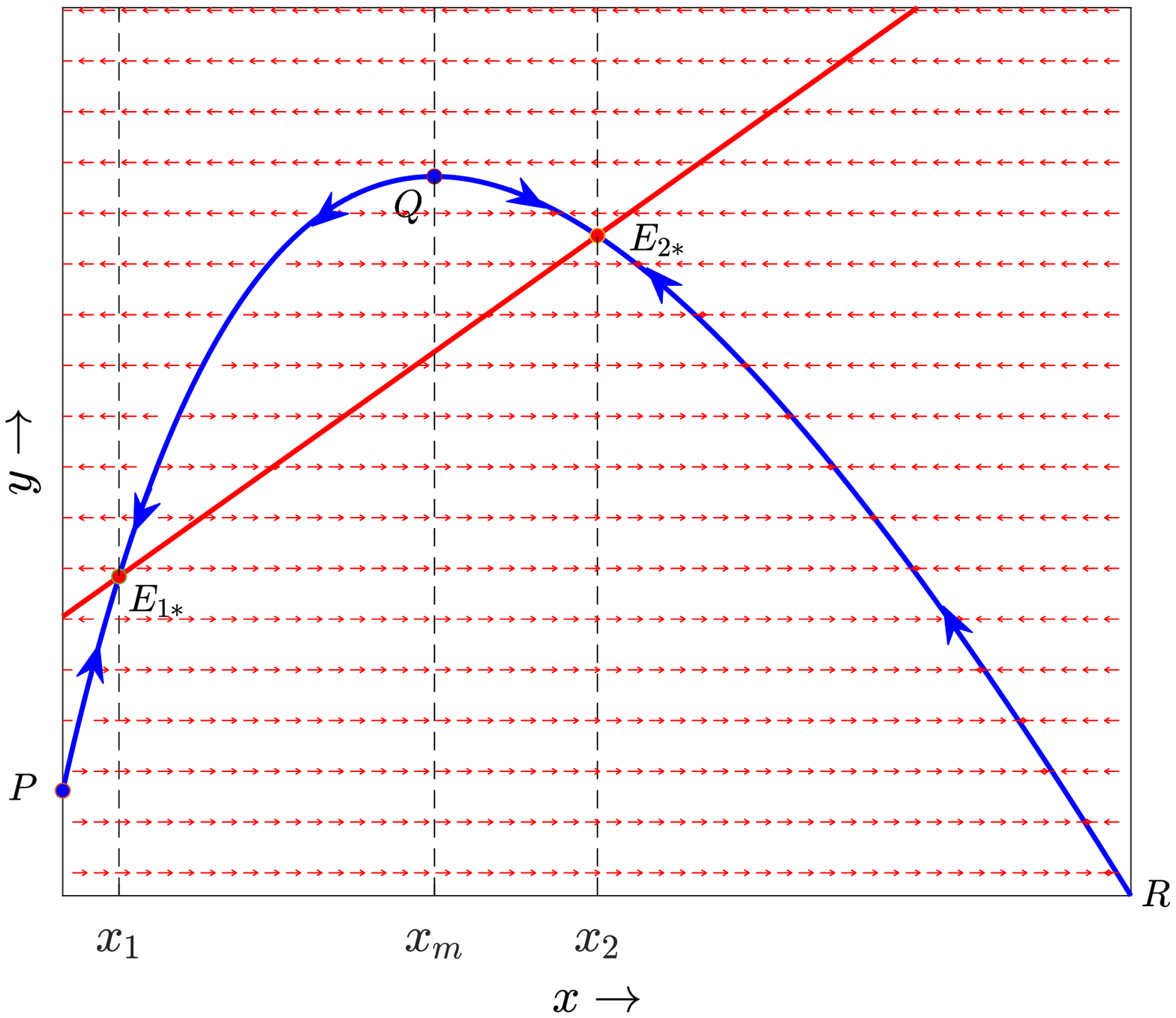}
\caption{The dynamics of the fast and slow subsystems \eqref{sf_fast subsystem} and  \eqref{sf_slow subsystem}, respectively, are illustrated. Two possible interior equilibrium positions are represented by solid red circles, and the non-hyperbolic points on the slow-manifold $M_{20}$ (blue curve) are shown by solid blue circles: the generic transcritical point $P(0,\beta/(\alpha-1) )$ and the generic fold point $Q(x_m, y_m)$. The normally hyperbolic attracting branch $S_0^a$ (from $Q$ to the point $R(1,0)$ for $x_m< x< 1$) and repelling branch $S_0^r$ (from $P$ to $Q$ for $0<x<x_m$) of the critical manifold $M_{20}$ are illustrated. The manifold $M_{10}$ is along the positive $y$-axis. The red arrows (horizontal) indicate fast flow, and the blue arrows on $M_{20}$ indicate slow flow. (For the interpretation of the colour references in this figure caption, the reader is referred to the web version of this chapter.)}
\label{fig:slow_fast_dynamics}			
\end{figure}

\section{Singular Hopf Bifurcation and Canard Cycles} \label{Sec.hopf_canard}
 Here, we assume $\mu\in R_1\cup R_3$ so that the existence of the interior equilibrium $E_{2*}$ ($E_{2*}=E_*$ for $\mu\in R_3$) is ensured. It also follows that for $\delta=\delta_*$, the interior equilibrium $E_{2*}$ coincides with the fold point $Q(x_m, y_m)$, where $\delta_*$ is implicitly given by the equation
\begin{align*}
    2(1+\theta)\left(1-\alpha+\sqrt{\alpha(\alpha-1-\beta)}\right)=\theta+1-\alpha\theta-\beta+\gamma-\delta+\sqrt{D},
\end{align*}
i.e.,
$$\delta_*=\left(\theta +2\right) \left(\alpha -\sqrt{\alpha  \left(\alpha -1-\beta \right)}-1\right)+\gamma -\beta +1.$$

We also observe that 
\begin{align*}
    \left.f(x,y)\right|_{(x_m,y_m,\delta_*)}&=0,\,\,\,\, \left.g(x,y)\right|_{(x_m,y_m,\delta_*)}=0,\\
    \left.\frac{\partial f(x,y)}{\partial x}\right|_{(x_m,y_m,\delta_*)}&=0,\,\,\,\,\left.\frac{\partial f(x,y)}{\partial y}\right|_{(x_m,y_m,\delta_*)}=-\frac{\alpha x_m(\beta+x_m)}{(\beta+x_m+y_m)^2}<0,\\
    \left.\frac{\partial g(x,y)}{\partial x}\right|_{(x_m,y_m,\delta_*)}&=\frac{\theta y_m^2}{(\delta_*+\theta x_m)^2},\,\,\,\,\left.\frac{\partial g(x,y)}{\partial \delta}\right|_{(x_m,y_m,\delta_*)}=\frac{y_m^2}{(\delta_*+\theta x_m)^2}=\frac{y_m^2}{(\gamma+y_m)^2}.
\end{align*}
Further, we assume that 
\begin{align}\label{canard condition}
    \left. \frac{\partial^2f}{\partial x^2}\right|_{(x_m,y_m,\delta_*)}&=2\left(\frac{\alpha y_m(\beta+y_m)}{(\beta+x_m+y_m)^3}-1\right)\neq 0.
\end{align}
Consequently, we have the following 
\begin{align}\label{canard condition 1}
  f(x_m,y_m,\delta_*)=0,\,\,g(x_m,y_m,\delta_*)=0,\,\,\text{and},\,\frac{\partial f}{\partial x}(x_m,y_m,\delta_*)=0,
\end{align}
and
\begin{align}\label{canard condition 2}
  \frac{\partial f}{\partial y}(x_m,y_m,\delta_*)\ne0,\,\,\frac{\partial^2 f}{\partial x^2}(x_m,y_m,\delta_*)\ne 0,\,\,\frac{\partial g}{\partial x}(x_m,y_m,\delta_*)\ne0 \,\,\, \frac{\partial g}{\partial \delta}(x_m,y_m,\delta_*)\ne0.
\end{align}
With the above assumption, the fold point $Q$ is now the non-degenerate canard point or the singular contact point of the system.  Using the transformation $X=x-x_m$, $Y=y-y_m$ and $\lambda=\delta-\delta_*$, the system \eqref{sf_model2} transforms to the following form

\begin{subequations}\label{sf_normal form}
\begin{align}
	\frac{dX}{dt}&=Y\left(a_{01}+a_{11}X+\mathcal{O}(X^2, XY, Y^2)\right)+X^2\left( \frac{a_{20}}{2}+\frac{a_{30}}{6}X+\mathcal{O}(X^2, Y^2)\right),\\
	\frac{dY}{dt}&=\epsilon\left[X\left(b_{10}+\frac{b_{20}}{2}X+b_{11}Y +\mathcal{O}(X^2, XY, Y^2)\right)+\lambda\left(\frac{y_m^2}{(y_m+\gamma)^2}+\mathcal{O}(X,Y,\lambda)\right)\right. \\    \nonumber
	&\left.+Y\left(b_{01}+\frac{b_{02}}{2}Y+\mathcal{O}(X^2, XY, Y^2)\right)\right],
\end{align}
\end{subequations} 
where $\displaystyle a_{ij}=\left.\frac{\partial^{i+j}f}{\partial u^i\partial v^j}\right|_{(u_m,v_m,\delta_*)}$ and $\displaystyle b_{ij}=\left.\frac{\partial^{i+j}g}{\partial u^i\partial v^j}\right|_{(u_m,v_m,\delta_*)}$ i.e.,

\begin{equation*}
  \begin{aligned}
  & a_{01}=-\frac{\alpha x_m(\beta+x_m)}{(\beta+x_m+y_m)^2},\,\, a_{20}=2\left(\frac{\alpha y_m(\beta+y_m)}{(\beta+x_m+y_m)^3}-1\right),\\
      & a_{11}=-\frac{\alpha}{(\beta+x_m+y_m)^3}\left(\beta(\beta+x_m+y_m)+2x_my_m)\right),
      \,\,a_{30}=-\frac{6\alpha y_m(\beta+y_m)}{(\beta+x_m+y_m)^4},\\
       & b_{10}=\frac{\theta y_m^2}{(\gamma+y_m)^2},\,\,b_{01}=-\frac{y_m^2}{(\gamma+y_m)^2},
       \,\,b_{20}=-\frac{2\theta^2 y_m^2}{(\gamma+y_m)^3},\,\, b_{02}=-\frac{2y_m(2\gamma+y_m)}{(\gamma+y_m)^3},\,\,
        b_{11}=\frac{2\theta y_m}{(\gamma+y_m)^2}.
    \end{aligned}
  \end{equation*}

In order to use the theory as developed in \cite{krupa2001relaxation}, we use the following re-scaling
\begin{align*}
    X=aX',\,\,Y=bY',\,\,\, t=ct'
\end{align*}
where
\begin{align*}
    a=-\frac{2b_{10}a_{01}}{a_{20}}\sqrt{-\frac{1}{a_{01}b_{10}}},\,\,b=\frac{2b_{10}}{a_{20}},\,\,\, c=\sqrt{-\frac{1}{a_{01}b_{10}}}.
\end{align*}
The system \eqref{sf_normal form} is then topologically equivalent to the following canonical form
\begin{subequations}\label{sf_standard normal form}
\begin{align}
	\frac{dX'}{dt'}&=-Y'h_1(X',Y')+X'^2 h_2(X', Y')+\epsilon h_3(X', Y'),\\
	\frac{dY'}{dt'}&=\epsilon\left[X'h_4(X',Y')-\lambda' h_5(X',Y',\lambda')+Y'h_6(X',Y')\right],
\end{align}
\end{subequations}
where
\begin{align*}
    h_1(X', Y')&=1-bca_{11}X'+\mathcal{O}(X'^2, X'Y', Y'^2),\\
    h_2(X', Y')&=1+\frac{a^2c}{6}a_{30}X'+\mathcal{O}(X'^2, Y'^2),\\
    h_3(X',Y')&=0,\,\,\, h_4(X', Y')=1+\frac{a^2c}{2b}b_{20}X'+acb_{11}Y'+\mathcal{O}(X'^2, X'Y', Y'^2),\\
    h_5(X',Y',\lambda')&=1+\mathcal{O}(X', Y', \lambda'),\\
    h_6(X', Y')&=cb_{01}+\frac{bc}{2}b_{02}Y'+\mathcal{O}(X'^2, Y'^2),    \lambda'=-\frac{c\lambda}{b}\frac{y_m^2}{(y_m+\gamma)^2}.
\end{align*}
Now, by the formulae (3.12) and (3.13) of \cite{krupa2001relaxation} we have 
\begin{subequations}\label{coefficients}
\begin{align}
   a_1&=\frac{\partial h_3}{\partial X'}(0,0)=0,\\
   a_2&=\frac{\partial h_1}{\partial X'}(0,0)=-bca_{11},\\
   a_3&=\frac{\partial h_2}{\partial X'}(0,0)=\frac{a^2c}{6}a_{30},\\
   a_4&=\frac{\partial h_4}{\partial X'}(0,0)=\frac{a^2c}{2b}b_{20},\\
   a_5&=h_6(0,0)=cb_{01},
   \end{align}
   \end{subequations}
   \begin{align}\label{criticality}
  \text{and}, A&=-a_2+3a_3-2a_4-2a_5=bca_{11}+\frac{a^2c}{2}a_{30}-\frac{a^2c}{b}b_{20}-2cb_{01}.
\end{align}

Hence, following the formulae (3.15) and (3.16) of \cite{krupa2001relaxation} the expansions of singular Hopf bifurcation and maximal canard curves are given by
\begin{align*}
    \lambda_H'(\sqrt{\epsilon})&=-\frac{a_1+a_5}{2}\epsilon+\mathcal{O}(\epsilon^{3/2}),\\
    \lambda_c'(\sqrt{\epsilon})&=-\left(\frac{a_1+a_5}{2}+\frac{A}{8}\right)\epsilon+\mathcal{O}(\epsilon^{3/2}).
\end{align*}
 In terms of original parameters, the singular Hopf and maximal canard curves can be written as
 
 \begin{align}\label{singular Hopf curve}
\delta_H(\sqrt{\epsilon})&=\delta_*+\frac{bb_{01}(y_m+\gamma)^2}{2y_m^2}\epsilon+\mathcal{O}(\epsilon^{3/2}),
\end{align}
\begin{align}\label{canard curve}
\delta_c(\sqrt{\epsilon})&=\delta_*+\frac{b(y_m+\gamma)^2}{4y_m^2}\left(b_{01}+\frac{b}{2}a_{11}+\frac{a^2}{4}a_{30}-\frac{a^2}{2b}b_{20}\right).
\end{align}

We assume $\mu_*= (\alpha, \beta, \gamma, \delta_*, \theta)$ so that for $\mu=\mu_*$, we have $\delta=\delta_*$.  Assuming $\mu_*\in R_1\cup R_3$ and the condition \eqref{canard condition}, we define a continuous family $\Gamma(s)$ of singular canard cycles for the vector field $V_{0,\mu_*}$ passing through the canard point $Q$ and consisting of a part of fast flow $y=s$ and parts of the attracting and repelling manifolds $S_{0}^a$ and $S_0^r$ as shown in Fig. \ref{fig:canard slow-fast cycles}, where $s\in (0, s_*)$ with 
\begin{align}\label{cycles}
s_*=\left\{\begin{array}{c} y_m-y_{1*},\,\,\,\, \mu_*\in R_1\\
y_m-\frac{\beta}{\alpha-1},\,\,\,\,\mu_*\in R_3\end{array}\right.
\end{align}

Assuming that $x_l(s)<x_r(s)$ be the two distinct roots of $\phi(x)=y_m-s$, we can parametrize the family of canard cycles $\Gamma(s)$ for $s\in (0, s_*)$ as follows.
\begin{align*}
\Gamma(s)=\left\{(x, \phi(x)): x\in [x_l(s), x_r(s)]\right\}\cup \left\{(x, y_m-s): x\in [x_l(s), x_r(s)]\right\}.
\end{align*}

The slow-fast cycle $\Gamma(s)$ as defined here is known as canard slow-fast cycle without head. In a similar fashion, assuming $\mu_*\in R_3$ and the condition \eqref{canard condition}, we define a continuous family of canard slow-fast cycles with a head $\bar{\Gamma}(s)$ for the vector field $V_{0,\mu_*}$ passing through the canard point $Q$ as follows (see Fig. \ref{fig:canard slow-fast cycles}).
\begin{align*}
\bar{\Gamma}(s)& =\left\{(x, y_m-s): x\in [0, x_l(s)]\right\}\cup \left\{(0, y_m-s): y\in [y', y_m-s]\right\}\cup \left\{(x, y'): x\in [0, x']\right\}\\
& \cup \left\{(x, \phi(x)): x\in [x_l(s), x']\right\},
\end{align*}
where $s\in \left(\frac{\beta}{\alpha-1}, \frac{2\beta}{\alpha-1}\right)$, $x'=\phi^{-1}(y')$ and $y'$ is defined by \eqref{relaxation_oscillation} in Lemma \ref{lemma_entry-exit}.
\begin{figure}[H]
\setlength{\belowcaptionskip}{-10pt}
\centering
\subfloat[]{\includegraphics[width=8cm,height=7cm]{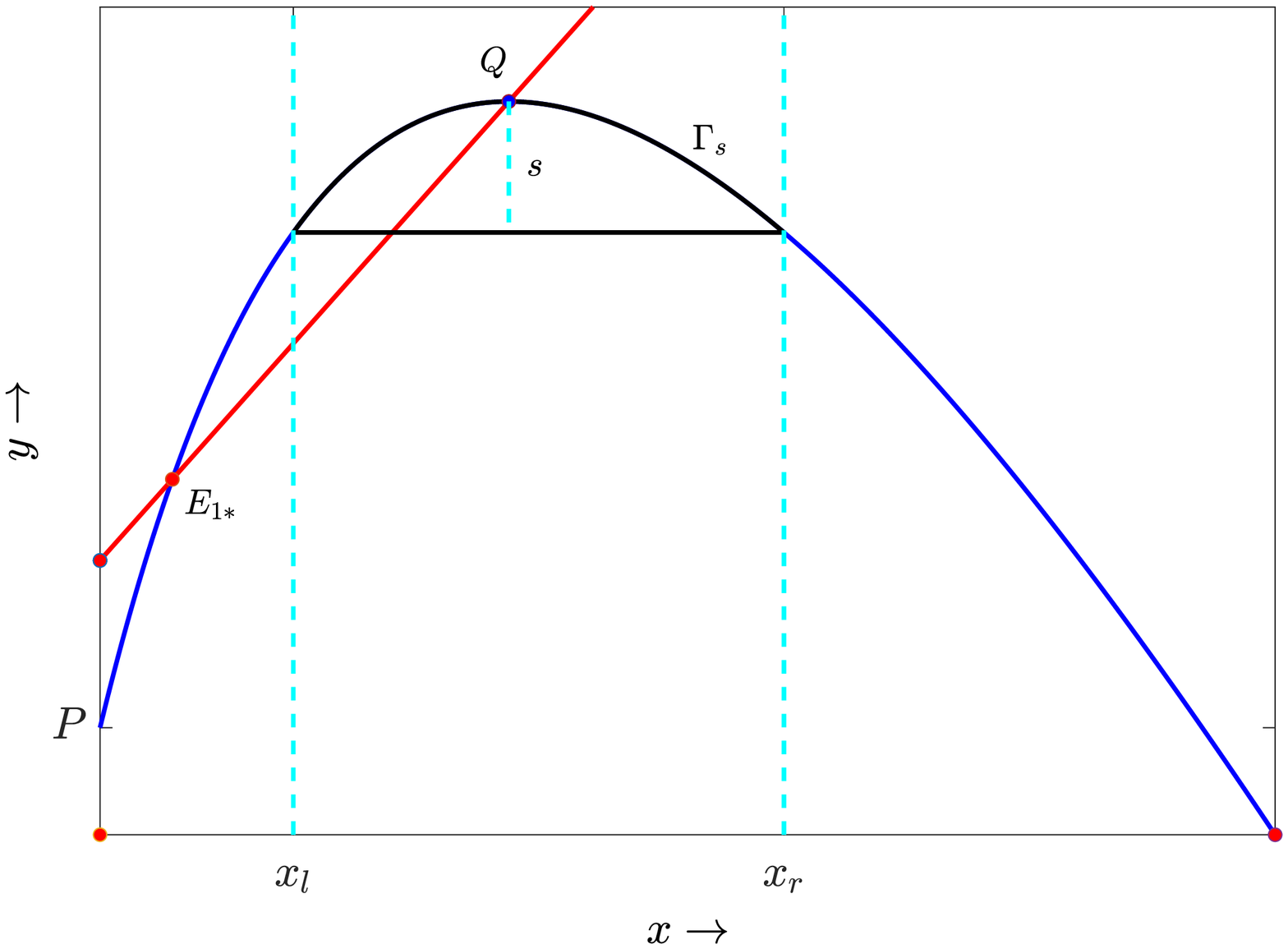}}
\subfloat[]{\includegraphics[width=8cm,height=7cm]{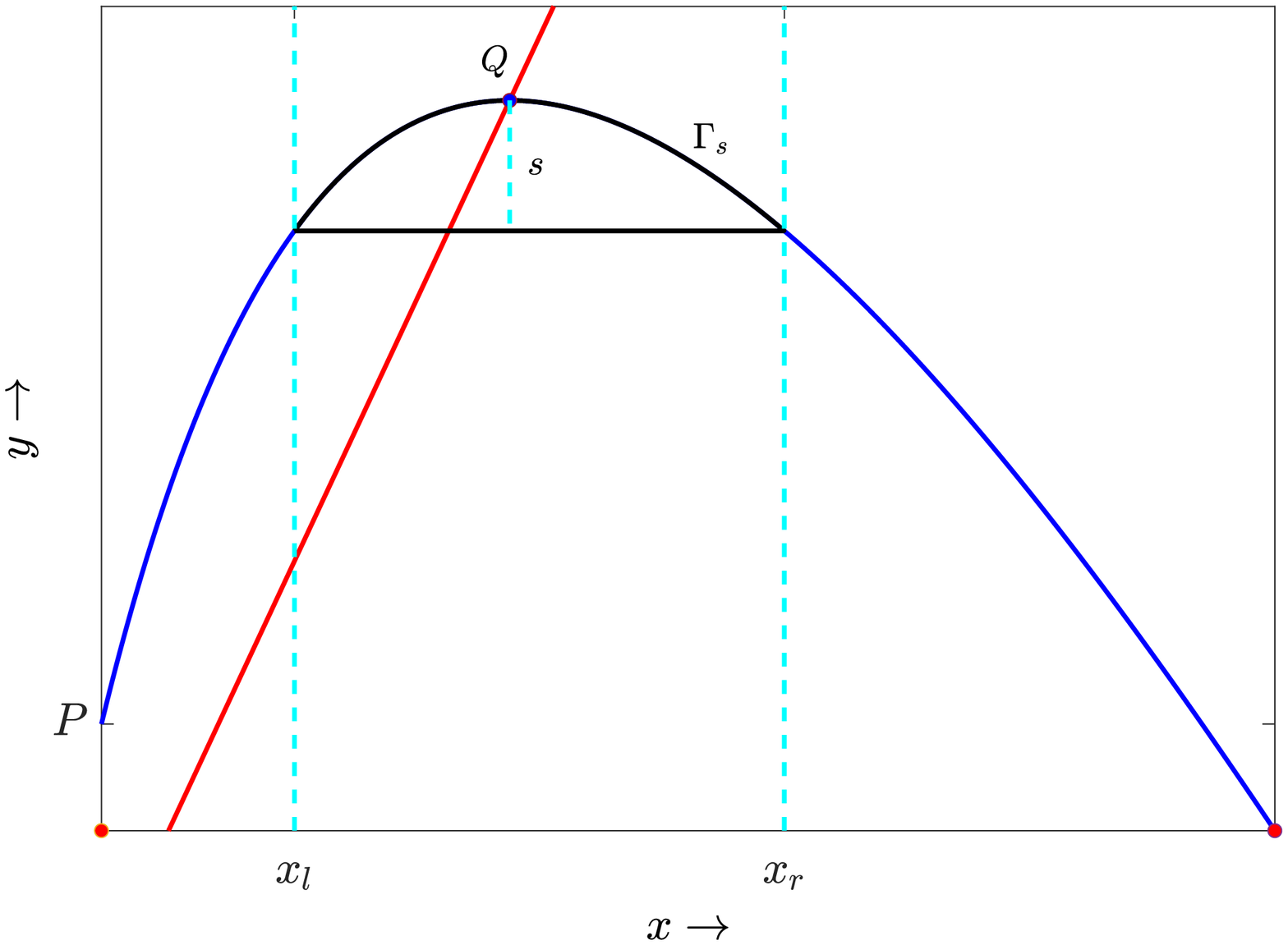}}\\
\subfloat[]{\includegraphics[width=8cm,height=7cm]{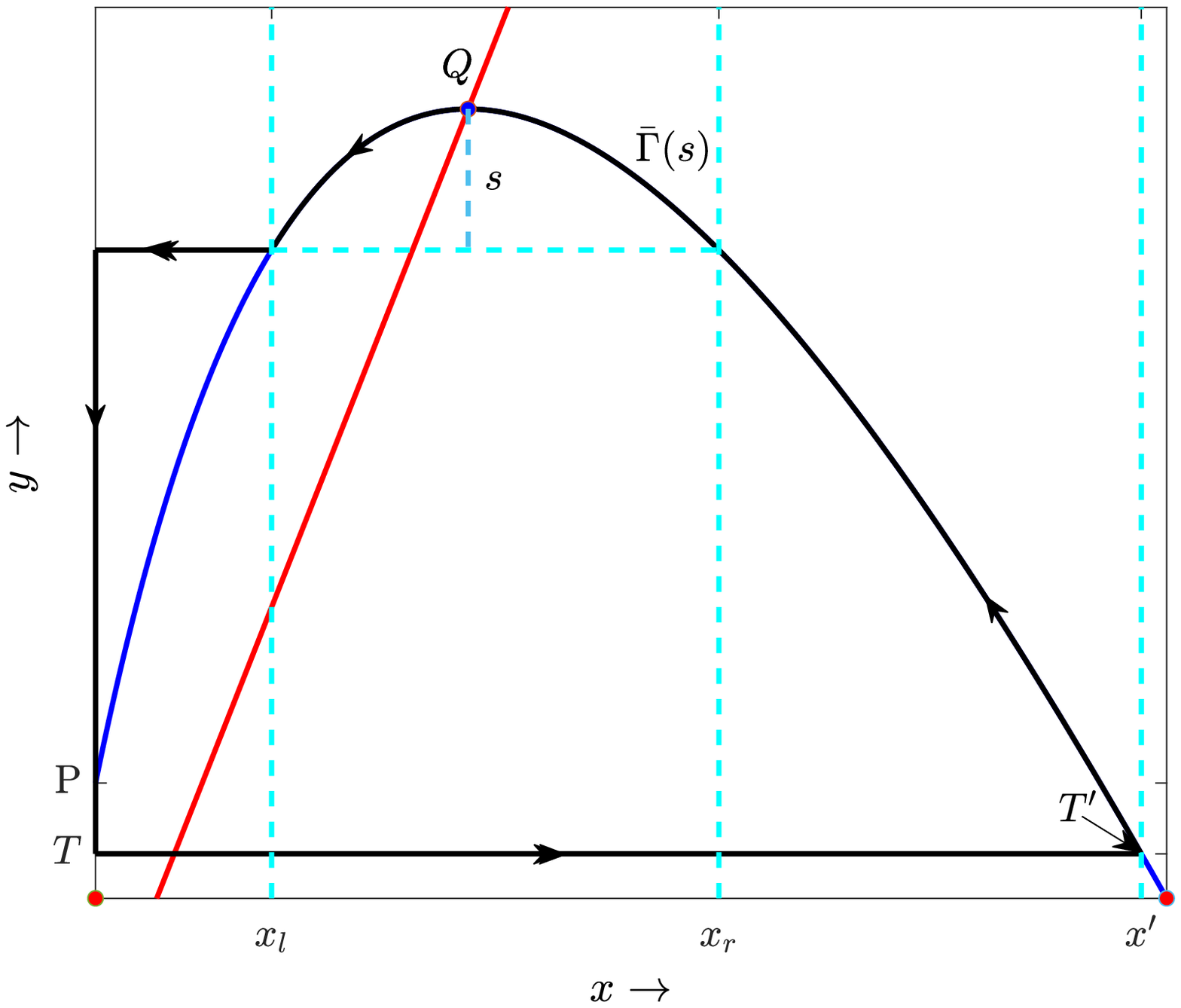}}
\caption{Slow-fast cycles without head through the canard point $Q(x_m, y_m)$ (black curve) of system \eqref{sf_model2} when (a) $\mu_*\in R_1$ (b) $\mu_*\in R_3$. (c) Slow-fast cycle with head through the canard point $Q(x_m, y_m)$ (black curve) of system \eqref{sf_model2}. The double arrows represent fast flow, and the single arrows represent slow flow.}
\label{fig:canard slow-fast cycles}			
\end{figure}

We now state the following results based on theorems (3.3) and (3.5) of \cite{krupa2001relaxation}.
\begin{theorem}\label{sf_singular_Hopf}
Assume $0<\epsilon\ll 1$, $\mu_*\in R_1\cup R_3$ and the condition \eqref{canard condition} hold. Then $\exists$ $\epsilon_0>0$ and $\delta_0>0$ such that for $0<\epsilon<\epsilon_0$ and $|\delta-\delta_*|<\delta_0$, the system \eqref{sf_model2} has an equilibrium point $Q_2$ in a neighbourhood of the fold point $Q$ which converges to $Q$ as $(\epsilon, \delta)\to (0,\delta_*)$. The system \eqref{sf_model2} undergoes a singular Hopf bifurcation at $\delta=\delta_H(\sqrt{\epsilon})$, where $\delta_H(\sqrt{\epsilon})$ is defined in \eqref{singular Hopf curve}. The Hopf bifurcation is non-degenerate when $A\neq 0$. It is supercritical if $A<0$ and sub-critical if $A>0$ where $A$ is given by \eqref{criticality}.
\end{theorem}

\begin{theorem}\label{canard explosion}
Assume $0<\epsilon\ll 1$, $\mu_*\in R_1\cup R_3$,  $K>0$ a constant and the condition \eqref{canard condition} hold. Then for every $s\in (0, s_*) (s\in \left(\frac{\beta}{\alpha-1}, \frac{2\beta}{\alpha-1}\right))$, the system \eqref{sf_model2} has a smooth family of canard cycles $s\to (\delta(s,\sqrt{\epsilon}), \Gamma(s,\sqrt{\epsilon})(\bar{\Gamma}(s,\sqrt{\epsilon}))$ bifurcating from the singular canard cycle $\Gamma(s)(\bar{\Gamma}(s))$ where $\delta(s, \sqrt{\epsilon})$ satisfies \begin{align*}
|\delta(s, \sqrt{\epsilon})-\delta_c(\sqrt{\epsilon})|\leq e^{-\frac{1}{\epsilon^K}},
\end{align*}
and $\delta_c(\sqrt{\epsilon})$ is given by \eqref{canard curve}. Moreover, $\Gamma(s,\sqrt{\epsilon})(\bar{\Gamma}(s,\sqrt{\epsilon}))$ approaches to $\Gamma(s)(\bar{\Gamma}(s))$ in the Hausdorff distance as $\epsilon\to 0$.
\end{theorem}

The phenomenon is manifested as follows: for $\mu_*\in R_1\cup R_3$ a supercritical singular Hopf bifurcation produces a small limit cycle, which quickly expands for the increase of the value of $\delta$. The shape of the cycle distorted during this expansion, and finally  a large amplitude stable oscillation is formed. This finding indicates that the instantaneous change from small to big cycles happen over an exponentially small parameter interval of $\delta$. This event is called a "canard explosion" (see Fig. \ref{fig:canard}b).

\begin{figure}[H]
\setlength{\belowcaptionskip}{-10pt}
\centering
	\subfloat[]{\includegraphics[width=8cm,height=7cm]{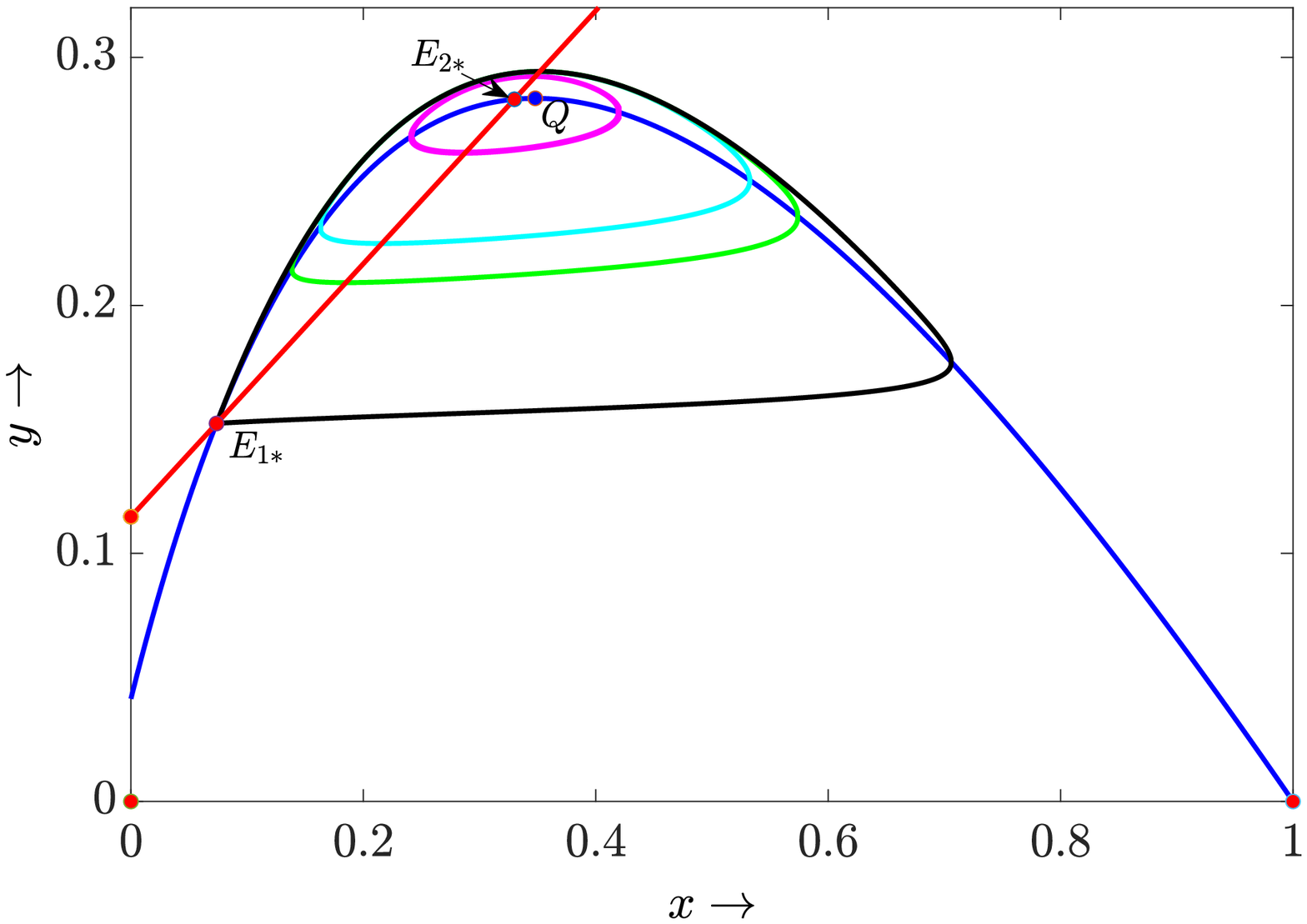}}
		\subfloat[]{\includegraphics[width=8cm,height=7cm]{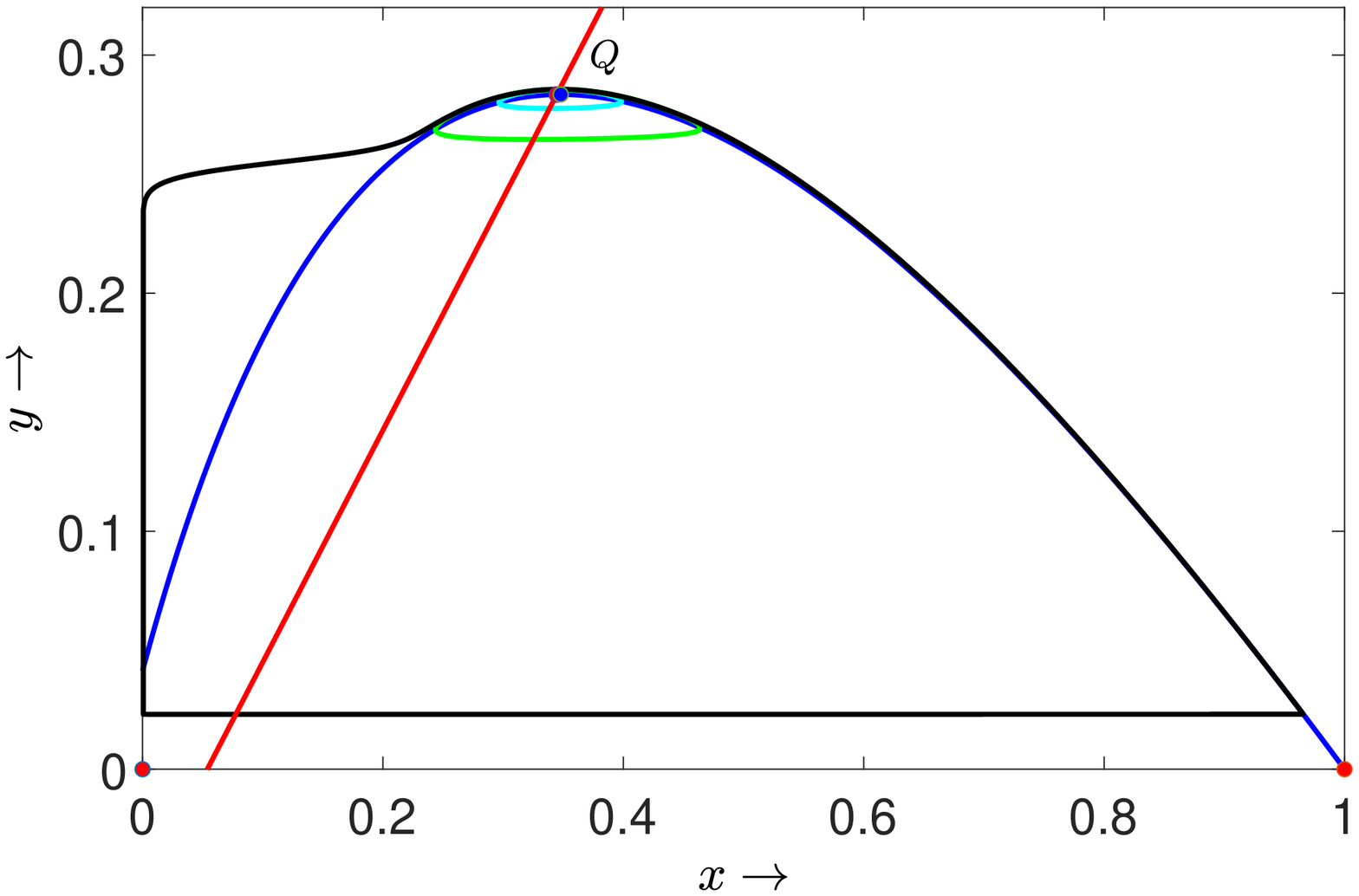}}
\caption{(a) For the case of $\mu\in R_1$, canard cycles and the birth of a homoclinic orbit at the canard point are illustrated. 
For four distinct values of $\delta$, namely $\delta=0.416$ (magenta), $0.0.4149$ (cyan), $0.41483$ (green), and $0.41481573598$ (black), four unstable periodic orbits are shown, with the amplitude of the orbits increasing with decreasing $\delta$. The figure evidently depicts the formation of a homoclinic orbit (black periodic orbit) via canard point for $\delta=0.41481573598$, whereby the homoclinic orbit connects the saddle equilibrium point $E_{1*}$. The other parameter values are $\alpha=1.5, \beta=0.0207, \gamma=0.3, \theta=0.51, \epsilon=0.1$. 
(b) For $\mu\in R_3$, the existence of the stable canard cycles and canard explosion phenomenon are seen in the diagram with small $\epsilon>0$.  For various values of $\delta$, such as $0.247$ (cyan), $0.24746$ (green), $0.24747$ (black), the stable periodic orbits (canard cycles) that arose through a supercritical singular Hopf bifurcation are displayed. When the bifurcation parameter $\delta$ is raised exponentially very small parameter interval from $0.24746$ (the green small amplitude cycle) to $0.24747$ (the black large amplitude canard cycle with a head), the figure shows that the amplitude of the orbit dramatically increases (canard explosion). The other parameter values are $\alpha=1.5, \beta=0.0207, \gamma=0.3, \theta=0.975,$ and  $\epsilon=0.01$.  (For interpretation of the references to colour in this figure caption, the reader is referred to the web version of this chapter.)} 
\label{fig:canard}			
\end{figure}

For $0<\epsilon\ll 1$, and $\mu\in R_3$, Fig. \ref{fig:bif_amplitude_max_canard}a depicts a simplified representation of the $\delta-\epsilon$ parametric plane that separates it into five distinct regions based on the locations of the threshold curves $\delta=\delta_H(\sqrt{\epsilon})$ (blue line), $\delta=\delta(s,\sqrt{\epsilon})$ (dashed black line),
$\delta=\delta_c(\sqrt{\epsilon})$ (red line), and $\delta=\delta_r(\sqrt{\epsilon})$ (dashed black line).
In domain $\raisebox{.5pt}{\textcircled{\raisebox{-.9pt} {1}}}$, when $\delta <\delta_H$, there exists no periodic orbit. After crossing the singular Hopf bifurcation threshold $\delta=\delta_H$ with the increase of $\delta$ and $0<\epsilon\ll 1$ fixed, as one moves from domain $\raisebox{.5pt}{\textcircled{\raisebox{-.9pt} {1}}}$ into domain $\raisebox{.5pt}{\textcircled{\raisebox{-.9pt} {2}}}$, the small-amplitude ($\mathcal{O}(\epsilon)$) stable periodic orbit develop for $\delta_H\sqrt{\epsilon}<\delta(s,\sqrt{\epsilon})$. The periodic orbit transforms into a canard cycle with or without a head when $\delta$ approaches the dashed line $\delta=\delta(s,\sqrt{\epsilon})$.
The size of the canard cycle increases on increasing $\delta$. Along the curve $\delta=\delta_r(\sqrt{\epsilon})$, the family of canard cycles ends at a stable relaxation oscillation (existence of relaxation oscillation has been shown in Sect. \ref{Sec:relaxation}) surrounding an unstable interior equilibrium point. The beginning and the end of the canard explosion are shown by the two black dashed lines $\delta=\delta(s,\sqrt{\epsilon})$ and $\delta=\delta_r(\sqrt{\epsilon})$ for sufficiently small $\epsilon$. Canard explosion describes the sudden change from a small canard cycle to a bigger relaxation oscillation within a limited range of the parameter $\delta$. In Fig. \ref{fig:bif_amplitude_max_canard}b, a numerical example is provided to better demonstrate this approach. It is clear from this numerical illustration that the amplitude of the periodic solution (the vertical axis) evolves more rapidly from small-amplitude canard cycles to large-amplitude relaxation oscillations when the governing parameter $\delta$ increases within an exponentially small range.

\begin{figure}[H]
\setlength{\belowcaptionskip}{-10pt}
\centering
	\subfloat[]{\includegraphics[width=8cm,height=7cm]{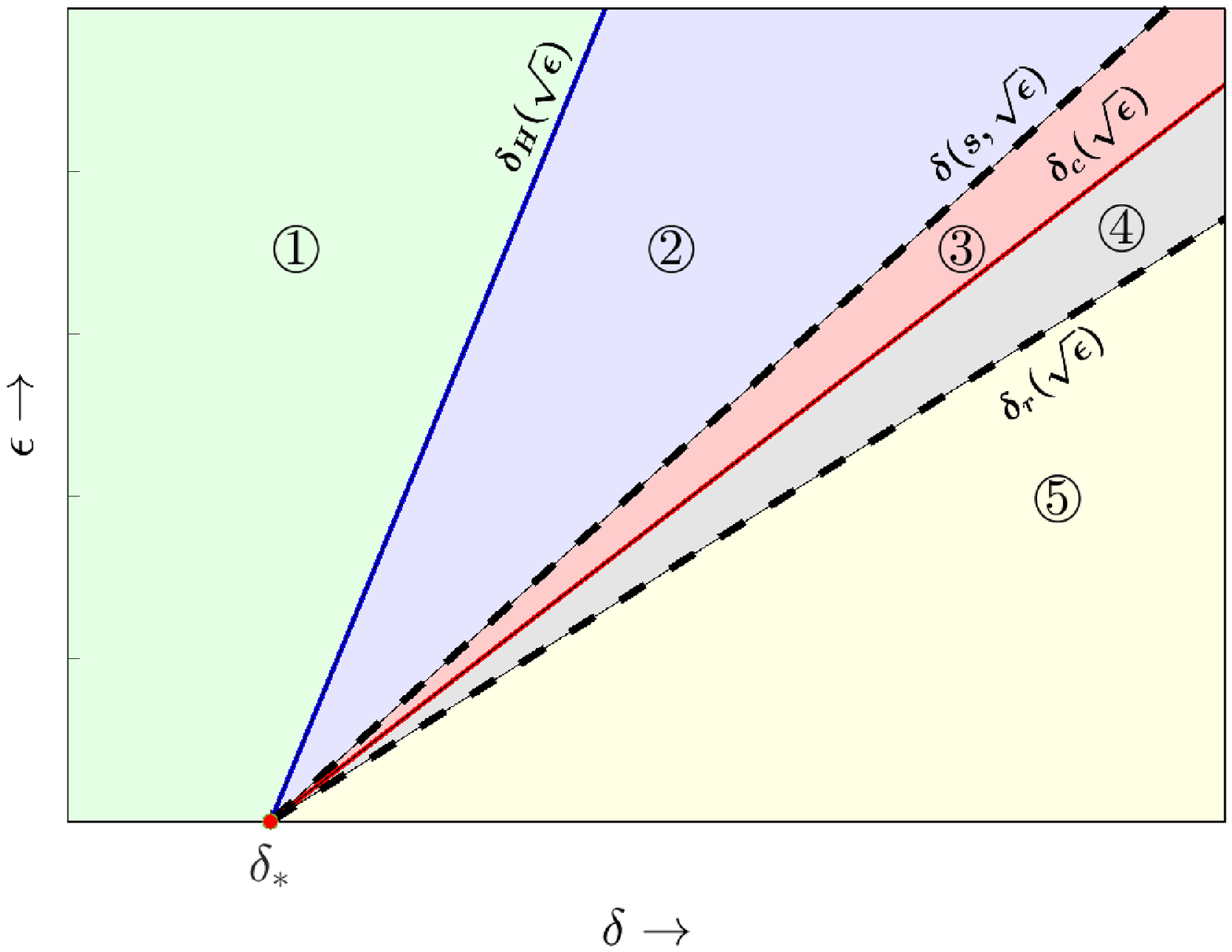}}\subfloat[]{\includegraphics[width=8cm,height=7cm]{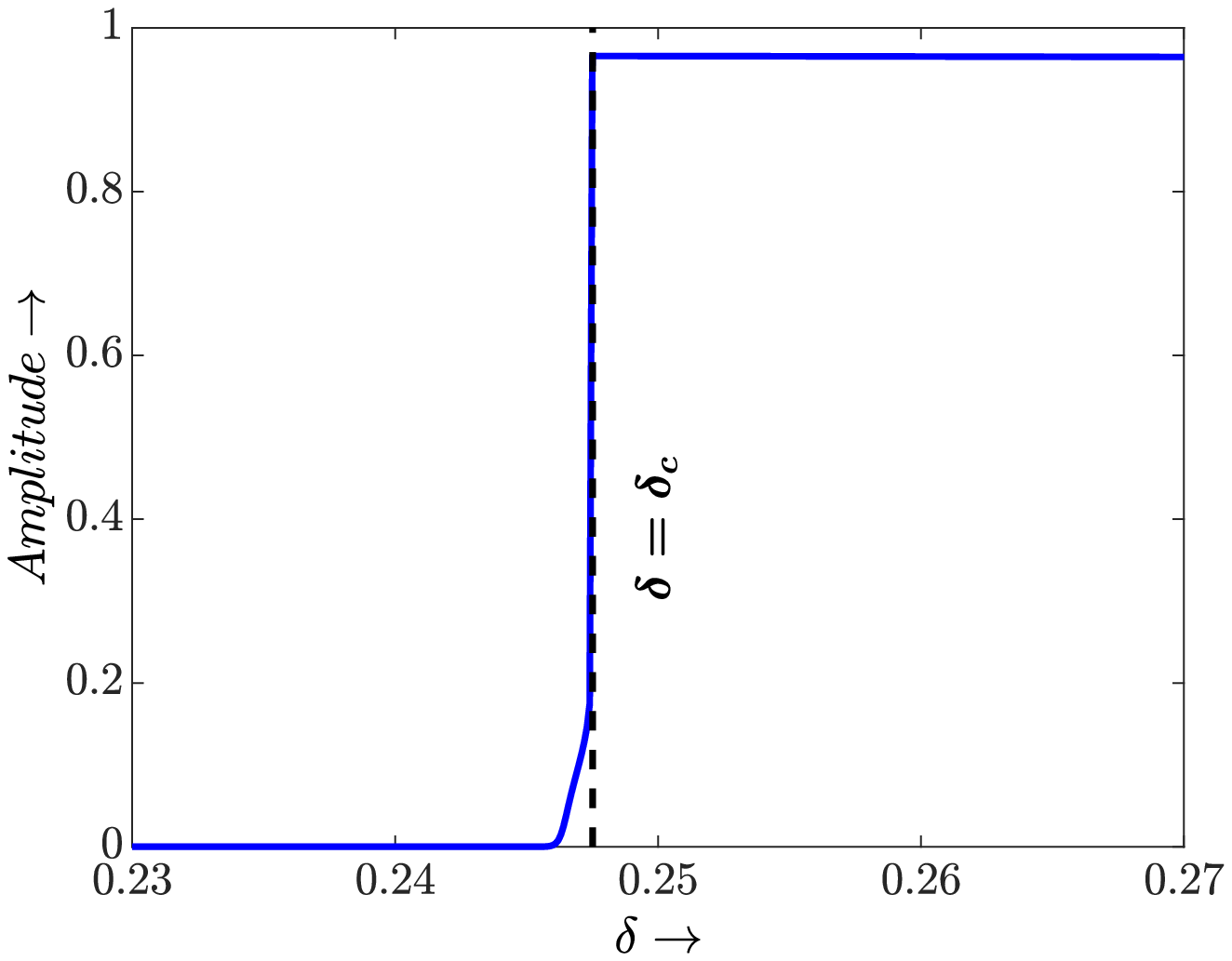}}
	\caption{(a) Schematic diagram depicting the singular Hopf bifurcation curve (blue), the maximum Canard curve (red), and the relaxation oscillation cycle (dashed black curve). (b)  A bifurcation diagram corresponding to the supercritical singular Hopf bifurcation for system \eqref{sf_model2} depicting the change in the amplitude of the canard cycles with respect to the variation of $\delta$ for fixed parameter values $\alpha=1.5$, $\beta=0.0207$, $\gamma=0.3$, $\epsilon=0.01$. Here $\delta_c=0.2475079340$. (For interpretation of the references to colour in this figure caption, the reader is referred to the web version of this chapter.)} 
\label{fig:bif_amplitude_max_canard}		
\end{figure}

Following the singular Hopf bifurcation, it follows that the Hopf-bifurcation threshold given by \eqref{singular Hopf curve} depends on $\epsilon$ and by computing the asymptotic expansion of the first Lyapunov coefficient in the blow-up coordinates \cite{kuehn2010from}, one can observe that the leading order term i.e., $A$ in the expansion of the first Lyapunov coefficient determines the criticality of the singular Hopf bifurcation with respect to $\epsilon\to 0$. Thus, the criticality of the singular Hopf bifurcation will be changed if $A$ changes its sign, i.e., the singular Hopf bifurcation will be degenerate if $A=0$. Using MATCONT, we uncover the occurrence of the codimension two generalized Hopf (GH) bifurcation. For the parameter values $\alpha=1.5$, $\beta=0.0207$ $\gamma=0.3$ and $\epsilon=0.01$, the generalized Hopf bifurcation (GH) point located at $(\delta_{GH}, \theta_{GH})=(0.361212, 0.638870)$ in the region $R_1$ (see Fig. \ref{fig:regions_bifurcation}), at which the critical point $(0.347909, 0.283481)$ has a pair of purely imaginary eigenvalues, $A=-2.2\times 10^{-8}$ i.e., $A$ is very near to zero and the leading order term of the second Lyapunov coefficient is positive. Consequently, the system undergoes a generalized Hopf bifurcation and the generalized Hopf bifurcation threshold is given by $(\delta_{GH}, \theta_{GH})=(0.361212 0.638870)$. Referring to the Fig. \ref{fig:regions_bifurcation}, the GH bifurcation occurs at the transition between supercritical ($H^-$) and sub-critical ($H^+$) Hopf bifurcations, the Hopf curve $H^+$ below the GH point is sub-critical whereas the Hopf curve above the GH is supercritical. The existence of a Limit Point of Cycles (LPC) curve has been detected emanating from the GH point propagates outward from the GH point towards the $H^-$ curve with a very close in distance with $H^-$ (see Fig. \ref{fig:regions_bifurcation}). It is to mention that in the region $R_1$ bounded by the LPC curve and the Hopf curve $H^{-}$, there exist two canard cycles (unstable and stable canard cycles). The unstable and the stable canard cycles collide and disappear via a saddle-node bifurcation of limit cycles on the LPC curve. It has been shown in Lemma \ref{lemma_5} that for the system parameters belonging to the region $R_1$ the prey-free equilibrium $E_{2b}$ is a hyperbolic stable node and correspondingly, when we have the existence of two canard cycles (stable and unstable) emerging due to the generalized Hopf bifurcation, the system exhibits a bi-stability phenomenon i.e., the system can either approach to ``prey extinction",  or ``oscillating coexistence" depending on the initial population size.

\section{Heteroclinic and Homoclinic Orbits}\label{Sec:hetro_homo_orbits}

In this section, we aim to show the existence of heteroclinic and homoclinic orbits for the system \eqref{sf_model2} under various parametric conditions.

\begin{proposition}\label{heteroclinic_1}
Assume $0<\epsilon\ll 1$, $\mu\in R_1$ and $x_{2*}>x_m$. Then we have the following results.
There exists one heteroclinic orbit connecting each pair of the equilibria $(E_0, E_{1b})$, $(E_0, E_{1*})$, $(E_{1b}, E_{2*})$, $(E_{1*}, E_{2b})$, $(E_{1*}, E_{2*})$ and infinitely many heteroclinic orbits connecting the pair of equilibria $(E_0, E_{2b})$, $(E_0, E_{2*})$. Moreover, the system \eqref{sf_model2} has no canard cycle and relaxation oscillation.

 \end{proposition}
 \begin{proof}

 For $\mu\in R_1$, the equilibria $E_0, E_{1b}$ are repelling and attracting saddle nodes, $E_{1*}$ is a hyperbolic saddle, $E_{2b}$ is a stable node and $E_{2*}$ is a stable equilibrium point. $E_0$  and $E_{1b}$ being saddle-node equilibria, a neighbourhood of them consists of two hyperbolic and one parabolic sectors. Every trajectory starting in $\IR_+^2$ and in a neighbourhood of $E_0$ move away from $E_0$ whereas, for $E_{1b}$, two hyperbolic sectors are separated by the two stable separatrices and an unstable separatrix. It also follows by Fenichel's theorem that for $0<\epsilon\ll 1$, the normally hyperbolic manifolds $S_0^r, S_0^a$ perturb to $S_{\epsilon}^{r}$ and $S_{\epsilon}^a$  which are within $\mathcal{O}(\epsilon)$ distance from $S_0^{r}$ and $S_0^a$ and the same scenario for the normally hyperbolic manifolds $S_0^{r+}$ and $S_0^{a+}$. The $x$ and $y$ axes are invariant under the flow and accordingly, the unstable orbit for $E_0$ along the positive $x$-axis gets connected to $E_{1b}$ forming a heteroclinic connection joining $E_0$ and $E_{1b}$.

$E_{1*}$ being a saddle equilibrium, the $\alpha$-limit set of one of its stable separatrix will be $E_0$ and hence, it forms a heteroclinic connection joining $E_0$ and $E_{1*}$. For  $0<\epsilon\ll 1$, the unstable separatrix of $E_{1b}$ follows $S_{\epsilon}^a$ slowly and gets attracted to $E_{2*}$ forming a heteroclinic orbit joining $E_{1b}$ and $E_{2*}$. One unstable separatrix of $E_{1*}$ is first attracted to the slow manifold $S_{\epsilon}^a$ in a fast timescale and then follows it slowly and, finally, gets attracted to the stable equilibrium $E_{2*}$, forming a heteroclinic connection joining $E_{1*}$ and $E_{2*}$. Another unstable separatrix of $E_{1*}$ is first attracted to the slow manifold $S_\epsilon^{a+}$ in a fast speed and then follows it in slow speed and finally gets attracted to the stable node $E_{2b}$, forming a heteroclinic orbit joining $E_{1*}$ and $E_{2b}$.

For the trajectories starting in the region bounded by the heteroclinic orbits joining the pair of equilibria $(E_0, E_{1b})$, $(E_0, E_{1*})$ and 
$(E_{1*}, E_{2*})$ the $\alpha$-limit set is $E_0$ and $\omega$-limit set is $E_{2*}$ as because all such trajectories will first get attracted to $S_{\epsilon}^a$ following approximately layers of the fast subsystem \eqref{sf_fast subsystem} and then follow $S_\epsilon^a$ in slow time and finally, attracted to the stable equilibrium $E_{2*}$. Hence, we have infinitely many heteroclinic orbits joining $E_0$ and $E_{2*}$. Similarly, all the trajectories starting in the region bounded by the heteroclinic orbits joining the pair of equilibria $(E_0, E_{1*})$ and $(E_0, E_{2b})$ have $E_0$ as their $\alpha$-limit set and $E_{2b}$ as their $\omega$-limit set and consequently, there exist infinitely many heteroclinic orbits joining $E_0$ and $E_{2b}$.

Under the parametric conditions as mentioned, the system has no canard point and no slow-fast cycle. Consequently, the system has no canard cycle and relaxation oscillation.

\end{proof}
\begin{figure}[H]
\setlength{\belowcaptionskip}{-10pt}
\centering
	\subfloat[]{\includegraphics[width=8cm,height=7cm]{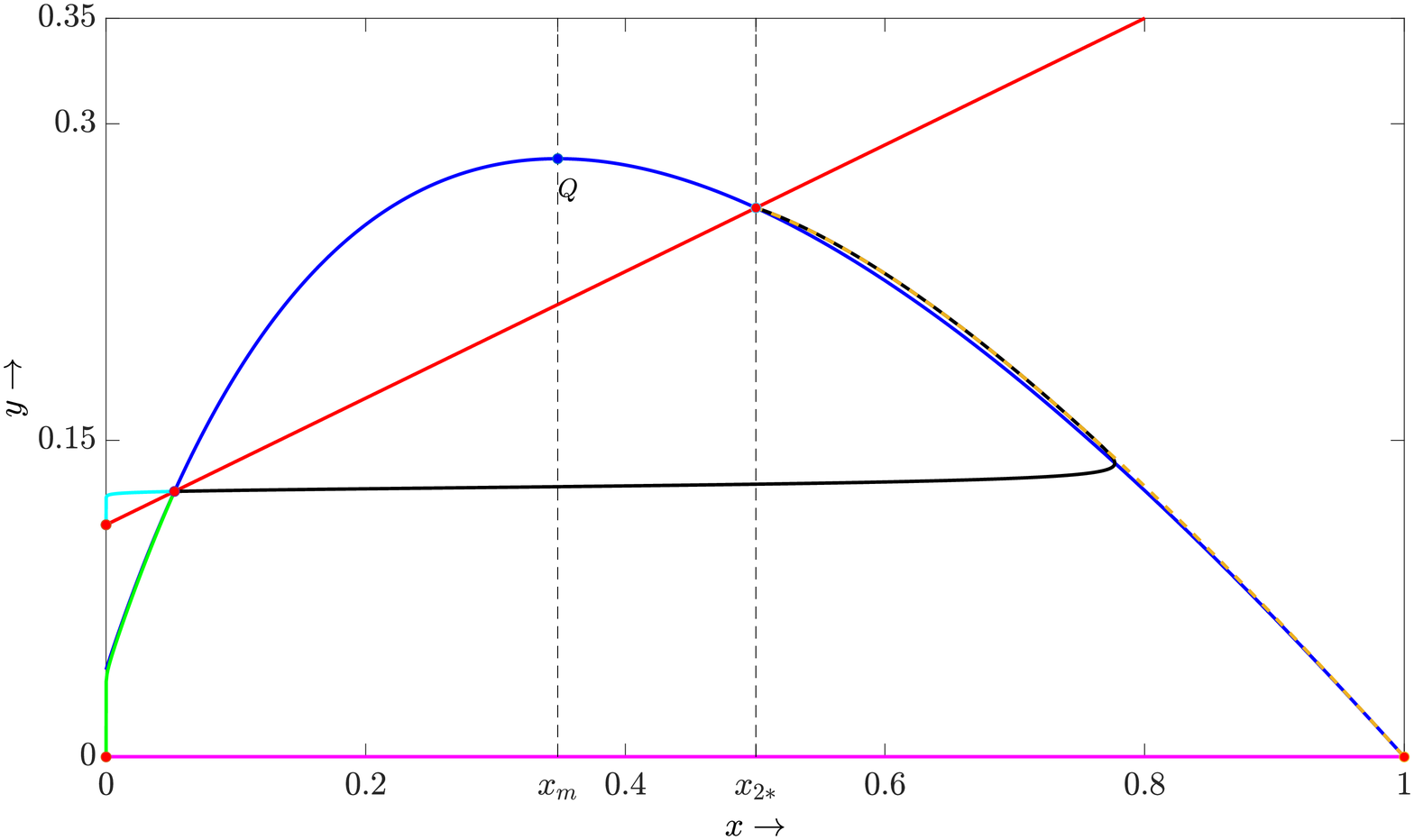}}
		\subfloat[]{\includegraphics[width=8cm,height=7cm]{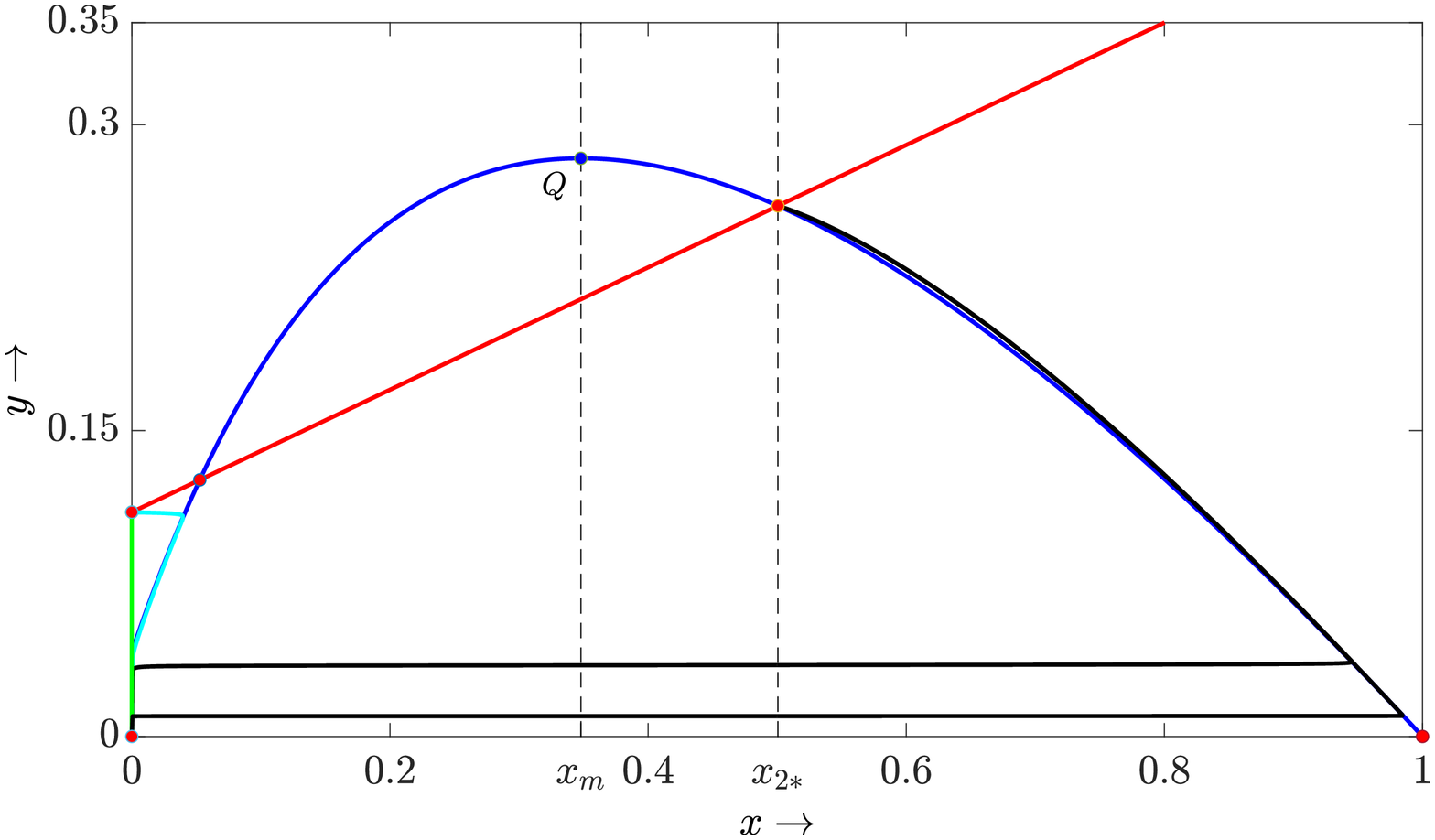}}
\caption{Numerical illustration for the existence of heteroclinic orbits for $\mu\in R_1$, $x_{2*}>x_m$, and $0<\epsilon\ll 1$. (a) Each set of equilibrium points $(E_0, E_{1b})$, $(E_0, E_{1*})$, $(E_{1b}, E_{2*})$, $(E_{1*}, E_{2b})$, and $(E_{1*}, E_{2*})$ is connected by a single heteroclinic orbit presented by solid magenta, solid green, dark yellow, solid cyan, and solid black curve respectively. (b) both green and cyan curves are the heteroclinic orbits connecting the equilibrium points $E_0$ and $E_{2b}$, and the two black curves are the 
 heteroclinic orbits connecting the equilibria $E_0$ and $E_{2*}$.  The parameter values are $\alpha=1.5, \beta=0.0207, \gamma=0.3,
 \delta=0.41, \theta=0.3, \epsilon=0.1$.  (For interpretation of the references to colour in this figure caption, the reader is referred to the web version of this chapter.)} 
\label{fig:heteroclinic_prop_1}			
\end{figure}

\begin{proposition}\label{homoclinic_1}
Assume $0<\epsilon\ll 1$, $\mu\in R_1$ and $x_{2*}=x_m$ and the condition \eqref{canard condition} hold. Then the system \eqref{sf_model2} has a unique homoclinic orbit connecting to the saddle equilibrium $E_{1*}$ if and only if $\delta=\delta_c(\sqrt{\epsilon})$, where $\delta_c(\sqrt{\epsilon})$ is given by \eqref{canard curve}. Furthermore, if this is the case then there exist infinitely many heteroclinic orbits connecting $E_0$ and $E_{2b}$, one heteroclinic orbit connecting each pair of the equilibria $(E_{1b}, E_{2b}), (E_0, E_{1*}), (E_0, E_{1b})$ and $(E_{1*}, E_{2b})$.
\end{proposition}

\begin{proof}
Under the said conditions $E_{1*}$ is a saddle, $E_{2*}$ is a canard point and the system \eqref{sf_model2} undergoes a singular Hopf bifurcation for $\delta=\delta_H(\sqrt{\epsilon})$. As in the Theorem \ref{heteroclinic_1}, by the Fenichel's theorem $S_0^r$ and $S_0^a$ perturb to the slow manifolds $S_\epsilon^r$ and $S_\epsilon^a$ respectively. One of the stable separatrices of $E_{1*}$, say, $W^s$ exactly follows, $S_\epsilon^r$ and the other stable separatrix of $E_{1*}$ has $E_0$ as its $\alpha$-limit set. Consequently, we have a heteroclinic orbit connecting $E_0$ and $E_{1*}$. For $ 0<\epsilon\ll 1$, one of the unstable separatrices  of $E_{1*}$, say $W^u$  first follows a layer of the fast subsystem \eqref{sf_fast subsystem} and then gets attracted to the slow manifold $S_{\epsilon}^a$ and finally, it reaches near the canard point $Q$ and by Theorem 3.2 in \cite{krupa2001relaxation} the slow manifolds $S_\epsilon^r$ and $S_\epsilon^a$ get connected for  $\delta=\delta_c(\sqrt{\epsilon})$ given by \eqref{canard curve}. This shows that for $\delta= \delta_c(\sqrt{\epsilon})$ given by \eqref{canard curve}, $W^s$ and $W^u$ get connected and form a homoclinic orbit homoclinic to $E_{1*}$.

Now, one of the unstable separatrices of $E_{1b}$ follows exactly $S_\epsilon^a$ and reaches in a neighbourhood of $S_\epsilon^{a+}$ passing the canard point $Q$ and finally, gets attracted to the stable node $E_{2b}$. Hence, we have a heteroclinic connection joining $E_{1b}$ and $E_{2b}$. Similarly, the other unstable separatrix of $E_{1*}$ follows a layer of the fast subsystem  \eqref{sf_fast subsystem} until it reaches in a neighbourhood of $S_\epsilon^{a+}$ and finally, gets attracted to the stable node $E_{2b}$ forming a heteroclinic orbit joining $E_{1*}$ and $E_{2b}$. 

Finally, following the same reason as in Proposition \ref{heteroclinic_1}, all the trajectories staring in $\IR^2_+$ and in a neighbourhood of $E_0$ except the heteroclinic connections joining $E_0$ and $E_{1b}$ along the $x$-axis and $E_0$ and $E_{1*}$ along the stable manifold of $E_{1*}$ have $E_{2b}$ as their $\omega$-limit set and $E_0$ as $\alpha$-limit set. Hence, we have infinitely many heteroclinic orbits joining $E_0$ and $E_{2b}$.
\end{proof}
\begin{figure}[H]
\setlength{\belowcaptionskip}{-10pt}
\centering
	\subfloat[]{\includegraphics[width=8cm,height=7cm]{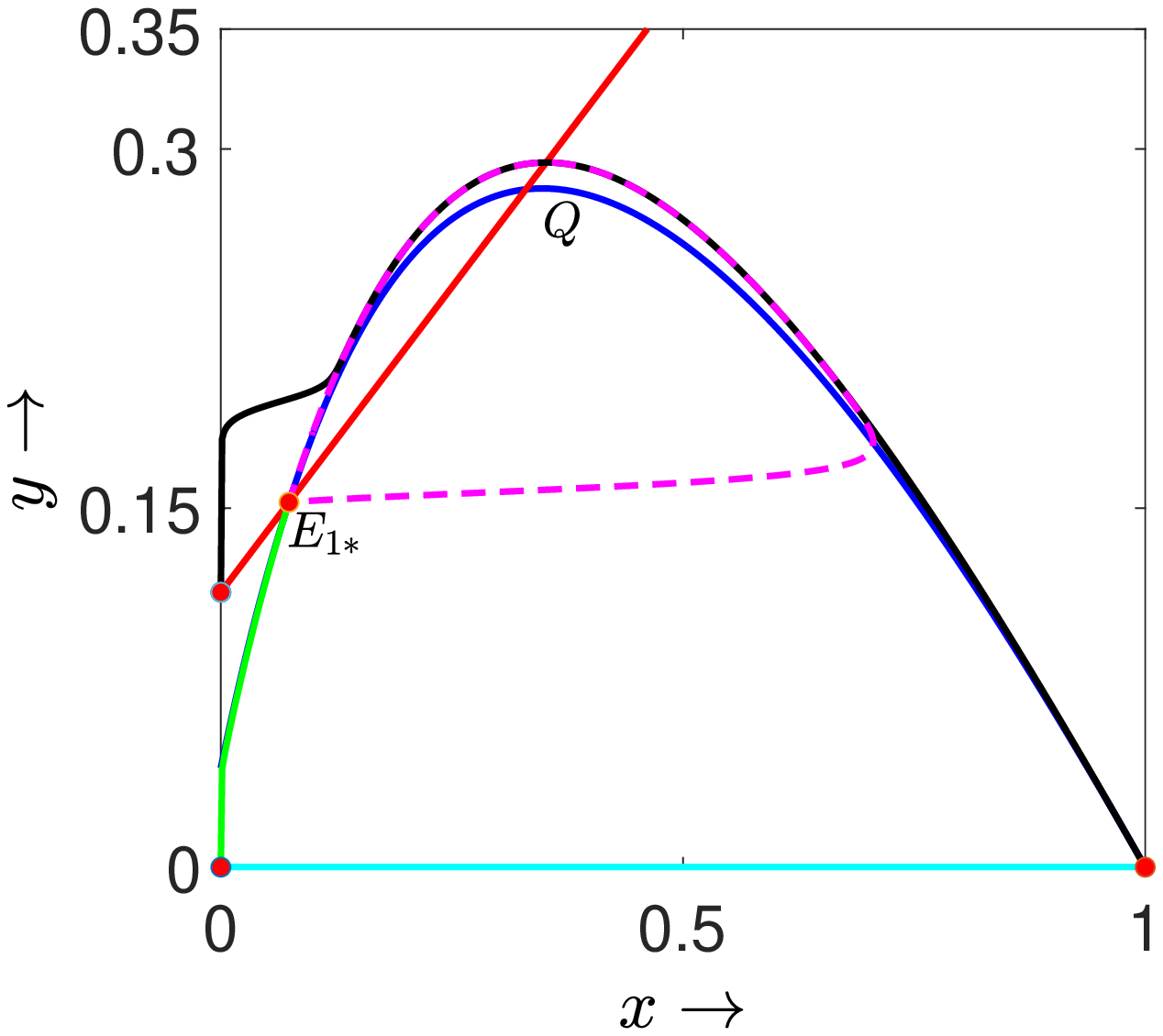}}
		\subfloat[]{\includegraphics[width=8cm,height=7cm]{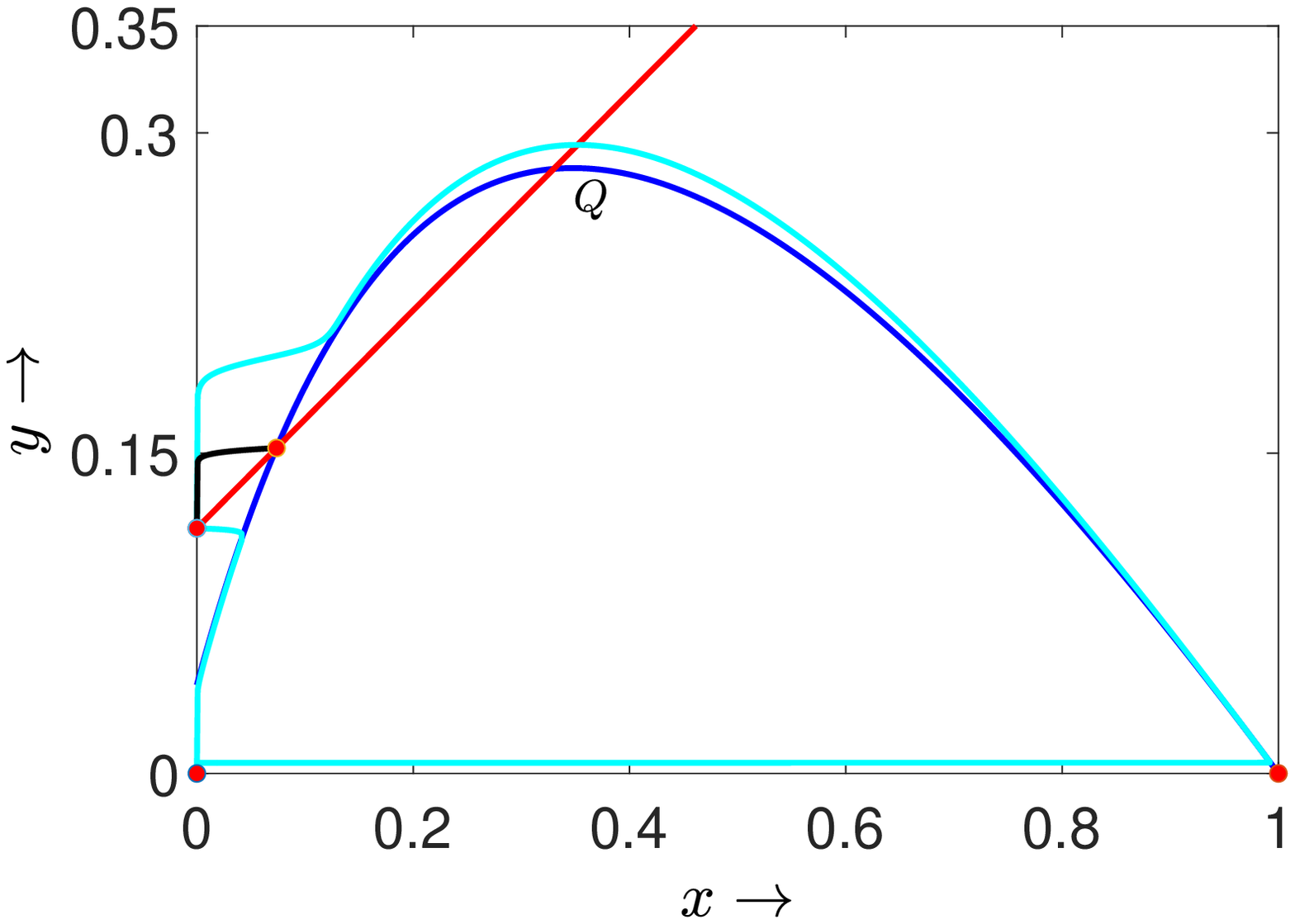}}
\caption{Numerical illustration for the existence of heteroclinic orbits for $\mu\in R_1$, $x_{2*}=x_m$, and $0<\epsilon\ll 1$. (a) Each set of equilibrium points $(E_{1b}, E_{2b})$, $(E_0, E_{1*})$, and $(E_0, E_{1b})$ is connected by a single heteroclinic orbit presented by solid black, solid green, and solid cyan curve respectively. The dashed magenta curve is the homoclinic orbit connecting the equilibrium point $E_{1*}$ to itself. (b) Both the cyan curves are the heteroclinic orbits connecting the equilibrium points $E_0$ and $E_{2b}$. The black curve is the unique heteroclinic orbit connecting the equilibria  $E_{1*}$ and $E_{2b}$.  The parameter values are $\alpha=1.5, \beta=0.0207, \gamma=0.3, \delta=0.41481573598, \theta=0.51, \epsilon=0.1$. (For interpretation of the references to colour in this figure caption, the reader is referred to the web version of this chapter.)} 
\label{fig:heteroclinic_prop_2}			
\end{figure}

\begin{proposition}\label{heteroclinic_2}
Assume $0<\epsilon\ll 1$, $\mu\in R_2$. Then the system \eqref{sf_model2} has one heteroclinic orbit connecting each pair of equilibria $(E_0, E_{1b}), (E_0, \bar{E}), (E_{1b}, E_{2b})$ and infinitely many heteroclinic orbits joining the pair of equilibria $(E_0, E_{2b})$ and $(\bar{E}, E_{2b})$. Furthermore, the system has neither canard cycle nor relaxation oscillation. 
\end{proposition}

\begin{proof}
Under the said condition the interior equilibria $E_{1*}$ and $E_{2*}$ get merged to the single equilibrium $\bar{E}$ which is saddle node in nature and a neighbourhood of $\bar{E}$ consists of two hyperbolic and one parabolic sectors (infinitely many centre manifolds and one unstable manifold). For $0<\epsilon\ll 1$, any orbit in the parabolic sector follows a layer of the fast subsystem \eqref{sf_fast subsystem} until it arrives in a neighbourhood of $S_\epsilon^a$ or $S_\epsilon^{a+}$. Now, if the orbit reaches a neighbourhood of $S_\epsilon^{a+}$, it is then attracted to the stable node, $E_{2b}$ forming a heteroclinic orbit connecting $\bar{E}$ and $E_{2b}$. Else, the orbit reaches in the vicinity of $S_\epsilon^a$, and following it passes the fold point by Theorem (2.1) of \cite{krupa2001relaxation}. The orbit then arrives in a $\mathcal{O}(\epsilon)$ neighbourhood of $S_\epsilon^{a+}$ following a layer of the fast subsystem \eqref{sf_fast subsystem} and finally, gets attracted to the stable node, $E_{2b}$ forming a heteroclinic connection joining $\bar{E}$ and $E_{2b}$. This phenomenon is true for all the orbits emanating from the parabolic sector of the saddle node $\bar{E}$. Hence, we have infinitely many heteroclinic orbits joining $\bar{E}$ and $E_{2b}$.

For $0<\epsilon\ll 1$, the unique stable branch of the infinitely many centre manifolds of $\bar{E}$ exactly follows $S_\epsilon^r$ and has $E_0$ as its $\alpha$-limit set. Therefore, the system \eqref{sf_model2} has a heteroclinic orbit connecting $E_0$ and $\bar{E}$. Proceeding in the same way as in the Proposition \ref{heteroclinic_2}, we have the existence of one heteroclinic orbit joining the pair of the equilibria $(E_0, E_{1b})$ and $(E_{1b}, E_{2b})$.

Finally, all the trajectories staring in $\IR^2_+$ and in a neighbourhood of $E_0$ except the heteroclinic connections joining $E_0$ and $E_{1b}$ along the $x$-axis and $E_0$ and $\bar{E}$ along the unique stable branch of the centre manifolds of $\bar{E}$ have $E_{2b}$ as their $\omega$-limit set and $E_0$ as $\alpha$-limit set. Consequently, we have infinitely many heteroclinic orbits joining $E_0$ and $E_{2b}$.

Under the said condition, the system has no canard point and no slow-fast cycle. Consequently, the system has no canard cycle and relaxation oscillation.
\begin{figure}[H]
\setlength{\belowcaptionskip}{-10pt}
\centering
	\subfloat[]{\includegraphics[width=8cm,height=7cm]{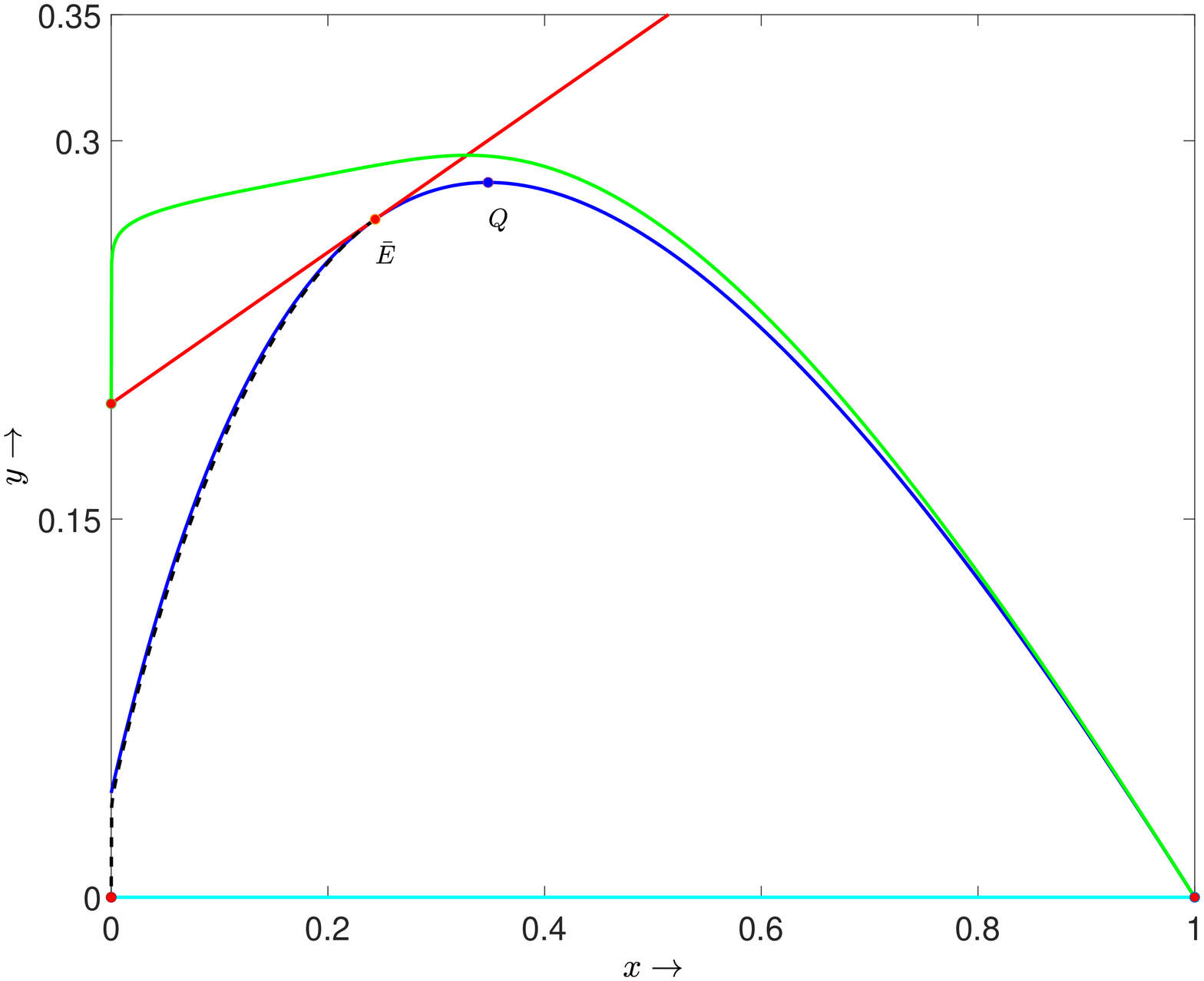}}
		\subfloat[]{\includegraphics[width=8cm,height=7cm]{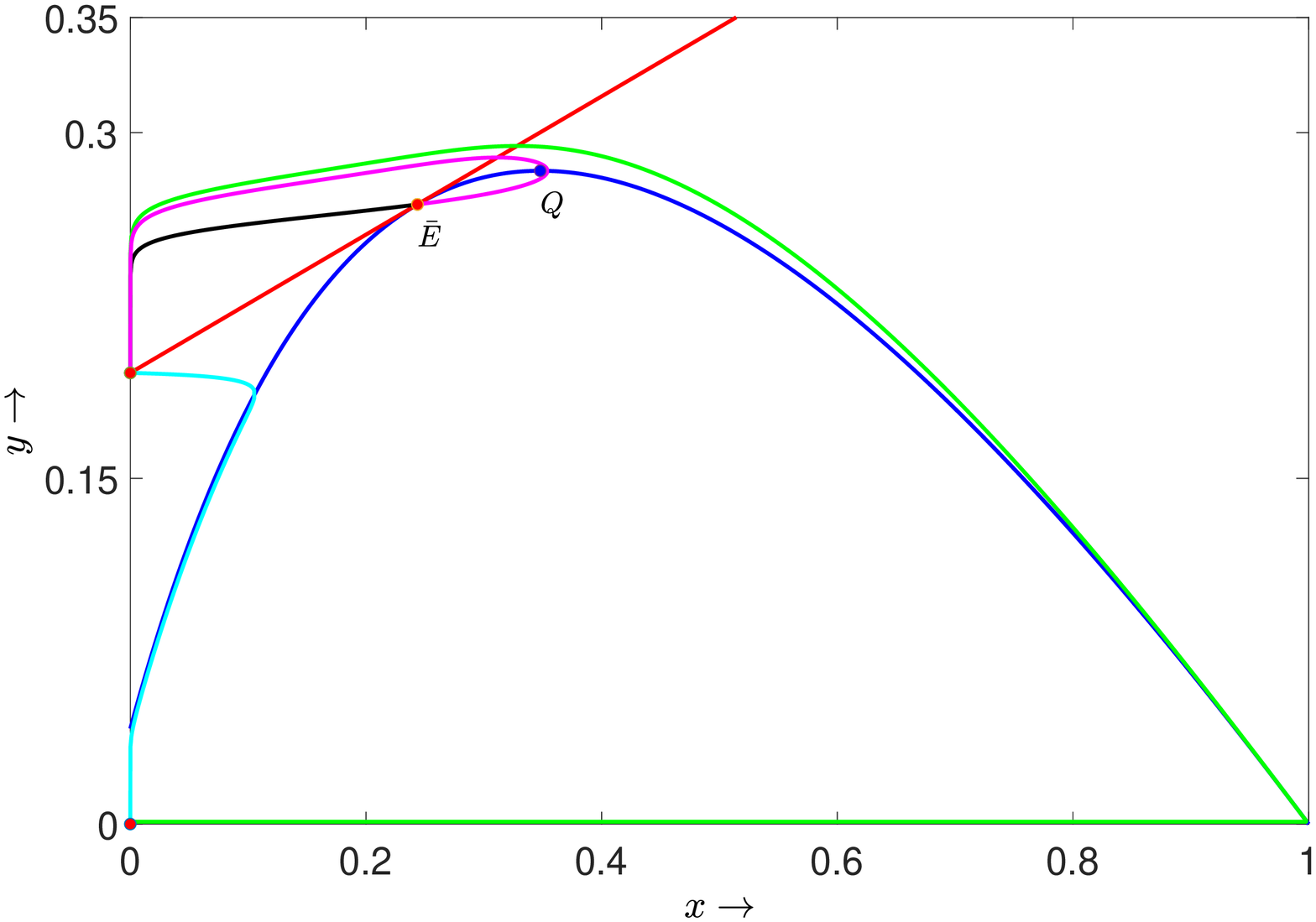}}
\caption{Numerical illustration for the existence of heteroclinic orbits for $\mu\in R_2$, and $0<\epsilon\ll 1$. (a) Each set of equilibrium points $(E_{0}, E_{1b})$, $(E_0, \bar{E})$, and $(E_{1b}, E_{2b})$ is connected by a single heteroclinic orbit presented by solid cyan, dashed black, and solid green curve respectively. (b) Both the magenta and black curves are the heteroclinic orbits connecting the equilibrium points $\bar{E}$ and $E_{2b}$. The cyan and green curves are the heteroclinic orbits connecting the equilibrium points $E_0$ and $E_{2b}$. The parameter values are $\alpha=1.5, \beta=0.0207, \gamma=0.3, \delta=0.49576955, \theta=0.3, \epsilon=0.1$. (For interpretation of the references to colour in this figure caption, the reader is referred to the web version of this chapter.)} 
\label{fig:heteroclinic_prop_3}			
\end{figure}
\end{proof}

\begin{proposition}\label{heteroclinic_3}
	Assume $0<\epsilon\ll 1$, $\mu\in R_3$ and $x_*>x_m$. Then the system \eqref{sf_model2} has a heteroclinic orbit connecting each pair of the equilibria $(E_0, E_{1b})$, $(E_0, E_{2b})$, $(E_{1b}, E_*)$, $(E_{2b}, E_*)$ and infinitely many heteroclitic orbits connecting $E_0$ and $E_*$. Moreover, if this is the case, then the interior equilibrium $E_*(x_*, y_*)$ is globally stable in the interior of $\IR^2_+$.
\end{proposition}

\begin{proof}
	For $0<\epsilon\ll 1$, the equilibrium $E_{2b}$ is a saddle whereas the unique interior equilibrium $E_*$ on the normally hyperbolic attracting manifold
	$S_0^a$ is a stable singularity. The stable manifold of $E_{2b}$ is exactly the critical manifold $M_{10}$ along the $y$-axis and as $y$-axis is invariant, there exists a heteroclinic connection between $E_0$ and $E_{2b}$. One of the unstable separatrices of $E_{2b}$ arrived in a neighbourhood of the slow manifold $S_\epsilon^a$ following a layer of the fast subsystem \eqref{sf_fast subsystem} in fast time and finally, following the slow manifold $S_\epsilon^a$ gets attracted to the stable singularity $E_*$. Hence, there exists a heteroclinic connection joining $E_{2b}$ and $E_*$. Proceeding as in the previous propositions, there also exists a heteroclinic connection joining the  equilibria $E_0$ and $E_{1b}$;  $E_{1b}$ and $E_*$. Finally, all the infinitely many centre manifolds in the parabolic sector of $E_0$ in $\IR^2_+$ have $E_{*}$ as their $\omega$-limit set and $E_0$ as $\alpha$-limit set. This shows that there exist infinitely many heteroclinic orbits connecting the singularities $E_0$ and $E_*$.
	
	One of the unstable separatrices of $E_{2b}$ first follows a layer of the fast subsystem \eqref{sf_fast subsystem} and reaches a neighbourhood of $S_\epsilon^a$, and following $S_\epsilon^a$ it gets attracted to the stable singularity $E_*$. From the geometry of the $S^a_0,$ it follows that $E_*$ is locally asymptotically stable. To claim that $E_*$ is globally stable, we need to show that there does not exist any periodic orbit in the interior of $\IR^2_+$. We consider a vertical line $x=x_m$ which divides the interior of $\IR^2_+$ into two domains $D_1$ and $D_2$, where 
	\begin{align*}
	D_1=\{(x,y):0<x\leq x_m, y>0\}\,\,\text{and}\,\,
	D_2=\{(x,y):x>x_m, y>0\}.
	\end{align*}
	In the domain $D_2$, we define the Dulac function $H:D_2\to\IR$ by $H(x,y)=\frac{\beta+x+y}{x(\alpha+x-1)y^2}$. Now, as in $D_2$, the critical manifold $M_{20}$ decreases, we have 
	\begin{align*}
	\frac{\partial (fH)}{\partial x} +\epsilon\frac{\partial(gH)}{\partial y}=
	\frac{1}{y^2}\phi'(x)-\frac{(\beta+x+y)}{x(\alpha+x-1)(y+\gamma)^2}<0,\,\, \forall\,\,(x,y)\in D_2.  
	\end{align*}
	Hence, by the Dulac criterion, the system \eqref{sf_model2} has no periodic orbit which entirely lies in $D_2$. Consequently, it follows that $E_{*}$ is the only $\omega$-limit point of every trajectory starting in $D_2$. For $0<\epsilon\ll 1$, we consider tracking of trajectories which start in $D_1$. This will suffice our claim if we can show that $E_*$ is the only $\omega$-limit point of the two trajectories $\Gamma_1$ and $\Gamma_2$ which start above and below the critical manifold $M_{20}$ but in $D_1$. The trajectory $\Gamma_1$ which starts in $D_1$ but above $M_{20}$ arrives in the vicinity of $S_\epsilon^{a+}$ as 
	$S_\epsilon^{a+}$ is hyperbolic attracting and passes the fold point $P$. Now, as the fold point $P$ is also a jump point, the trajectory $\Gamma_1$ then moves away from the normally hyperbolic repelling manifold $S_\epsilon^{r+}$ following a fast layer of the subsystem \eqref{sf_fast subsystem} and reaches a $\mathcal{O}(\epsilon)$ neighbourhood of $S_\epsilon^{a}$ which lies in the region $D_2$. Finally, the trajectory gets attracted to $E_*$ following $S_\epsilon^{a}$, i.e., $E_*$ is the $\omega$-limit point of the trajectory $\Gamma_1$. Similarly, the trajectory $\Gamma_2$ starting below $M_{20}$ but in $D_1$ arrives in the vicinity of $S_\epsilon^{a}$ following a fast layer of the subsystem \eqref{sf_fast subsystem} and gets attracted to $E_*$ following $S_\epsilon^{a}$. Thus, any trajectory starting in $D_1$ or $D_2$ converges to $E_*$ showing that $E_*$ is globally stable for $0<\epsilon\ll 1$.

\begin{figure}[H]
\setlength{\belowcaptionskip}{-15pt}
\centering
	\subfloat[]{\includegraphics[width=8cm,height=6cm]{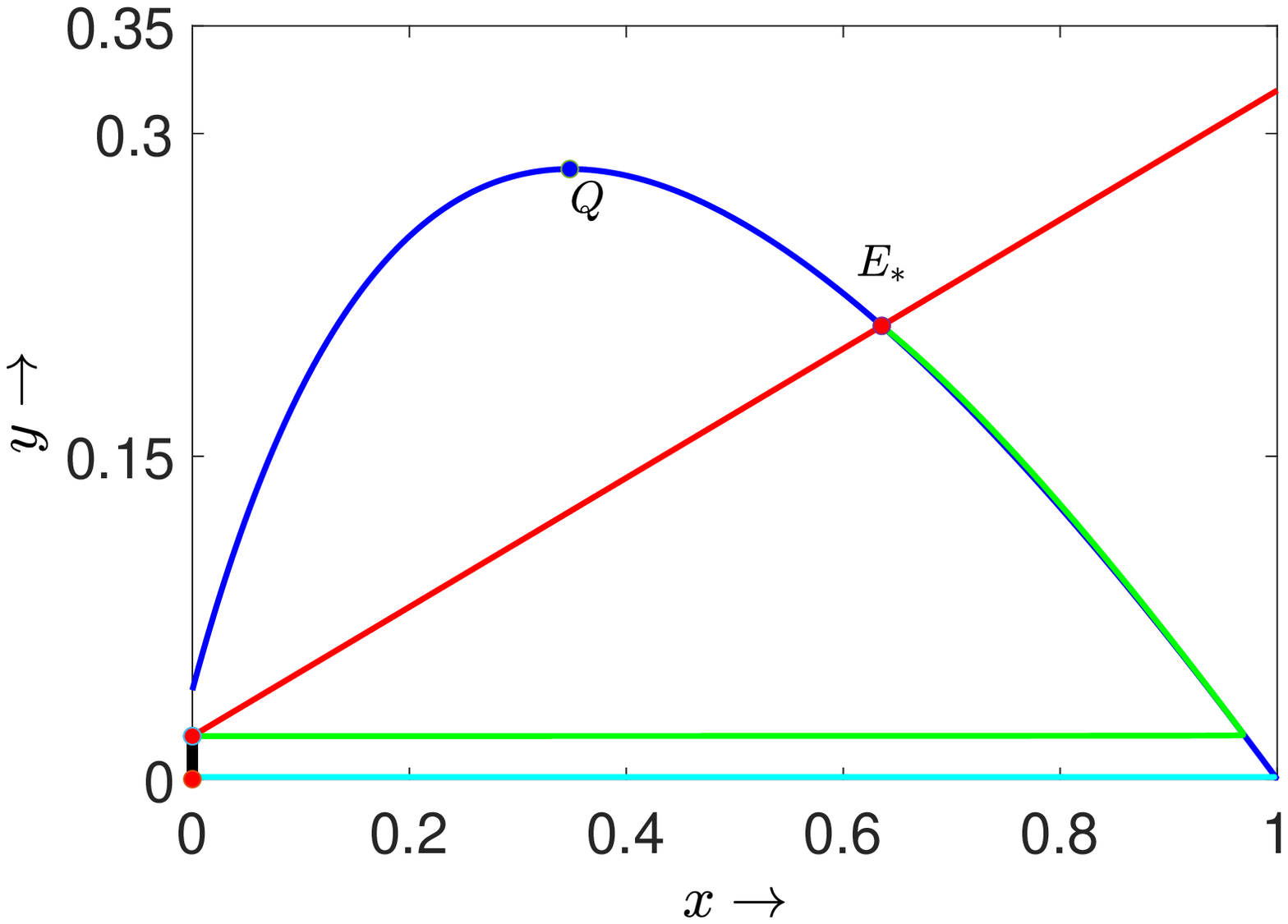}}
		\subfloat[]{\includegraphics[width=8cm,height=5.8cm]{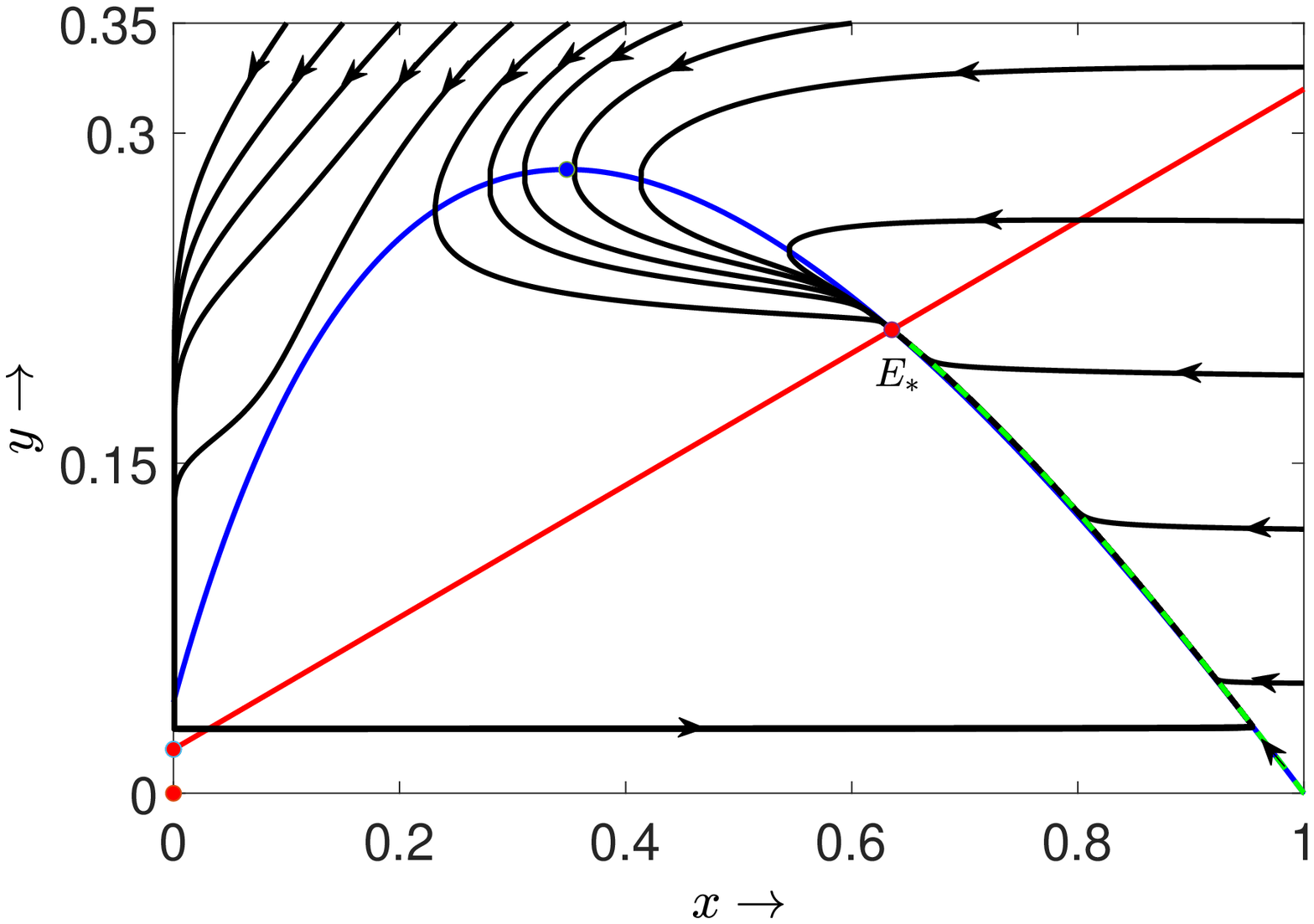}}\\
  \subfloat[]{\includegraphics[width=8.5cm,height=6cm]{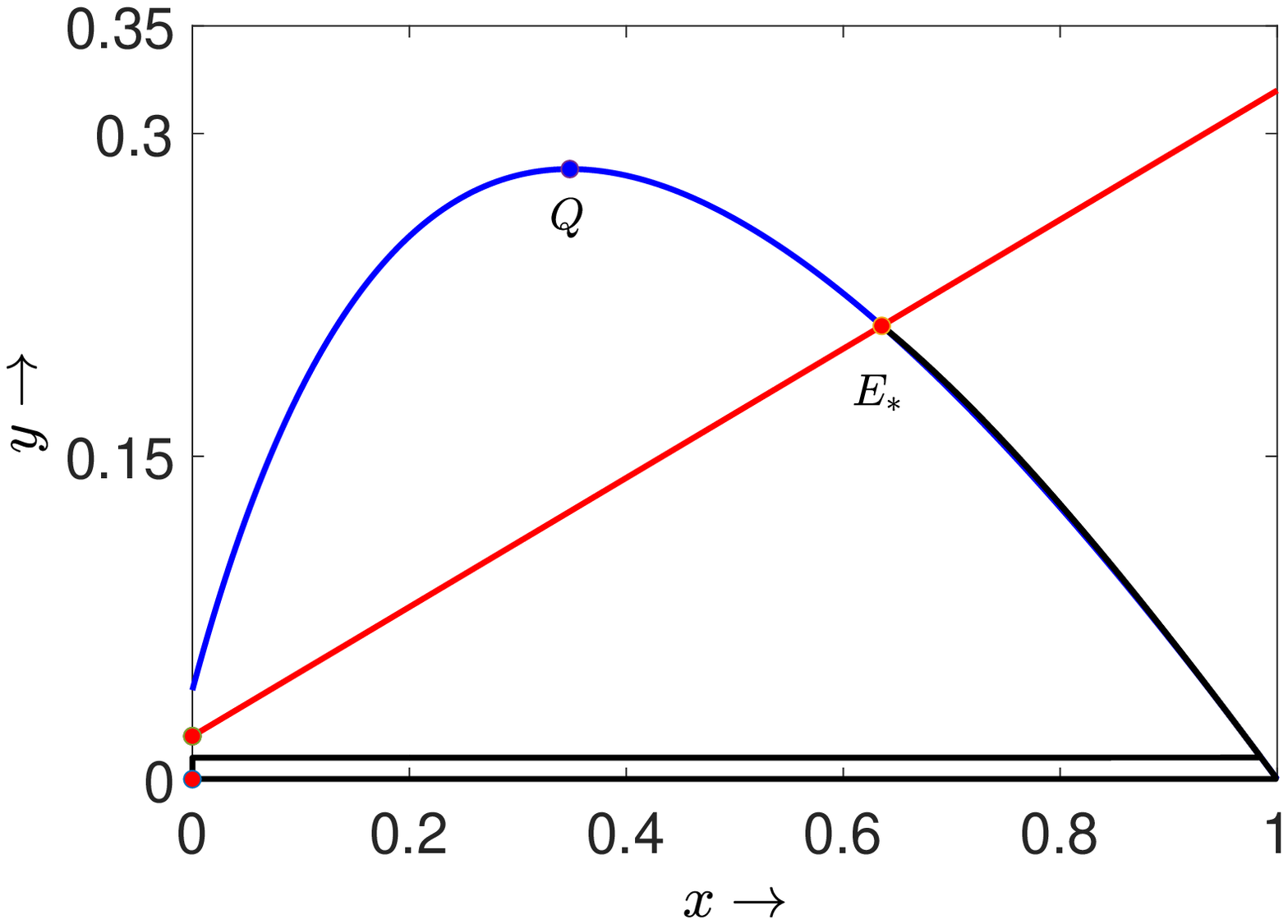}}
\caption{Numerical illustration for the existence of heteroclinic orbits for $\mu\in R_3$, $x_{*}>x_m$ and $0<\epsilon\ll 1$. (a) Each set of equilibrium points $(E_{0}, E_{1b})$, $(E_0, E_{2b})$ and $(E_{2b}, E_{*})$ is connected by a single heteroclinic orbit presented by solid cyan, black and green curve respectively. (b) The equilibrium point $(E_{1b}, E_{*})$ is connected by a unique heteroclinic orbit presented by broken green. $E_*$ is globally asymptotically stable, and there does not exist any periodic orbit in the first quadrant. (c)
Then the system \eqref{sf_model2} has infinitely many heteroclitic orbits (using black curves, we only present two here) connecting $E_{0}$ and $E_*$. The parameter values are $\alpha=1.5, \beta=0.0207, \gamma=0.3, \delta=0.32, \theta=0.3, \epsilon=0.1$. (For interpretation of the references to colour in this figure caption, the reader is referred to the web version of this chapter.)} 
\label{fig:heteroclinic_prop_4}			
\end{figure}
\end{proof}

\begin{proposition}
	Assume $0<\epsilon\ll 1$, $\mu\in R_3$ and $x_*<x_m$. Then the system \eqref{sf_model2} has a unique stable relaxation oscillation.
\end{proposition}

\begin{proof}
The proof has been shown in the next  section with the help of entry-exit function.
\end{proof}

\section{Relaxation Oscillation}\label{Sec:relaxation}
Here, our target is to show the existence of relaxation oscillation for the system \eqref{sf_model2} for $0<\epsilon\ll 1$ whenever $\mu\in R_3$ and $x_{2*}<x_m$ with the help of entry-exit function. A relaxation oscillation for the system \eqref{sf_model2} is a periodic orbit $\Gamma_\epsilon$ which converges to a piece-wise smooth singular closed orbit $\Gamma_0$ consisting of slow fast segments as $\epsilon\to 0$ in the Hausdorff distance. 

We know that the critical manifold $M_{10}$ i.e., the $y$-axis is normally hyperbolic attracting for $y>\frac{\beta}{\alpha-1}$ and normally hyperbolic repelling for $y<\frac{\beta}{\alpha-1}$. We consider the system \eqref{sf_model2} and observe that for $\epsilon=0$, the $y$- axis consists of equilibria, attracting for $y>\frac{\beta}{\alpha-1}$ and repelling for $y<\frac{\beta}{\alpha-1}$. For $\epsilon>0$, very small, a trajectory starting at $(x_0,y_0)$, $x_0>0$, very small, $y_0>\frac{\beta}{\alpha-1}$ gets attracted towards the $y$-axis and then drifts downward and when cross the line $y=\frac{\beta}{\alpha-1}$ gets repelled from the $y$-axis. Thus, for $\epsilon>0$, very small, the trajectory re-intersects the line $x=x_0$ at $(x_0, p_\epsilon(y_0))$ such that $\displaystyle\lim_{\epsilon\to 0}p_\epsilon(y_0)=p_0(y_0)$, where $p_0(y_0)$ is determined by 
\begin{align}
\int_{y_0}^{p_{0(y_0)}}\frac{1-\frac{\alpha y}{\beta}}{y^2\left(\frac{1}{y+\gamma}-\frac{1}{\delta}\right)}dy=0.
\end{align}
The function $y_0\to p_0(y_0)$ is said to be an entry-exit function.
\begin{lemma}\label{lemma_entry-exit}

If $\gamma-\delta\geq 0$ or $\gamma-\delta< 0$ then there exists a unique $y'$, where $0<y'<\frac{\beta}{\alpha-1}$ or $\delta-\gamma<y'<\frac{\beta}{\alpha-1}$ such that 
\begin{align}\label{relaxation_oscillation}
J(y')=\int_{y'}^{y_m} \frac{1-\frac{\alpha y}{\beta+y}}{y^2\left(\frac{1}{y+\gamma}-\frac{1}{\delta}\right)}dy=0.
\end{align}
\end{lemma}

\begin{proof} We have, 
\begin{align*}
J(y)&=\int_{y}^{y_m} \frac{1-\frac{\alpha y}{\beta+y}}{y^2\left(\frac{1}{y+\gamma}-\frac{1}{\delta}\right)}dy\\
&= -\delta\int_y^{y_m}\frac{(\beta+(1-\alpha) y)(y+\gamma)}{y^2(\beta+y)(y+\gamma-\delta)}dy\\
&=-\delta\left[\frac{\alpha\gamma(\delta-\gamma)-\beta\delta}{\beta(\gamma-\delta)^2}\int_y^{y_m}\frac{1}{y}dy+\frac{\gamma}{\gamma-\delta}\int_y^{y_m}\frac{1}{y^2}dy+\frac{\alpha(\beta-\gamma)}{\beta(\beta-\gamma+\delta)}\int_y^{y_m}\frac{1}{\beta+y}dy\right.\\
&\left.+\frac{\delta(\alpha-1)(\gamma-\delta)+\beta\delta)}{(\beta-\gamma+\delta)(\gamma-\delta)^2}\int_y^{y_m}\frac{1}{y+\gamma-\delta}dy\right]\\
&\rightarrow -\infty \,\, {\rm as}\,\, y\to 0^+, \, \gamma-\delta\geq 0\,\,{\rm or}\,\,{\rm as}\,\, y\to(\delta-\gamma)^+,\, \gamma-\delta<0.
\end{align*}

Further,
\begin{align*}
J'(y)=\delta\int_y^{y_m}\frac{(\beta-(\alpha-1) y)(y+\gamma)}{y^2(\beta+y)(y+\gamma-\delta)}>0, 
\end{align*}
${\rm either \,\,for}\,\,0<y<\frac{\beta}{\alpha-1},\,\gamma-\delta\geq 0\,\, {\rm or\,\, for }\,\,\delta-\gamma<y<\frac{\beta}{\alpha-1},\, \gamma-\delta<0$.
Hence, $J(y)$ increases strictly for $0<y<\frac{\beta}{\alpha-1}$, $\gamma-\delta\geq 0$ or for $\delta-\gamma<y<\frac{\beta}{\alpha-1}$, $\gamma-\delta<0$.

We also have,
\begin{align*}
J\left(\frac{\beta}{\alpha-1}\right)=
-\delta\int_{\frac{\beta}{\alpha-1}}^{y_m}\frac{(\beta-(\alpha-1) y)(y+\gamma)}{y^2(\beta+y)(y+\gamma-\delta)}dy>0. 
\end{align*}
Thus, it follows that there exists a unique $y'$ where $0<y'<\frac{\beta}{\alpha-1}$, $\gamma-\delta\geq 0$ or $\delta-\gamma<y'<\frac{\beta}{\alpha-1}$, $\gamma-\delta<0$ such that $J(y')=0$.
\end{proof}
The critical manifolds $M_{10}$ and $M_{20}$ lose its normal hyperbolicity at $P(0,\frac{\beta}{\alpha-1})$ and $Q(x_m, y_m)$. The point $Q(x_m,y_m)$ is a generic fold point for the system \eqref{sf_model2} and also a jump point as at this point the fast flow \eqref{sf_fast subsystem} is moved away from the critical manifold $M_{20}$ and gets attracted toward $S_0^{a+}$, the attracting branch of the critical manifold $M_{10}$. For the point $P$, we have 
\begin{align*}
\left.\frac{\partial f}{\partial x}\right|_{P}= 0 =\left.\frac{\partial f}{\partial y}\right|_{P}, ~~\left.\frac{\partial^2 f}{\partial x^2}\right|_{P}=2\frac{\alpha(1-\beta)-1}{\alpha\beta}>0, ~~ \left.g\right|_P=\frac{\beta^2}{(\alpha-1)^2}\left(\frac{(\delta-\gamma)(\alpha-1)-\beta}{
\delta(\beta+\gamma(\alpha-1))}\right)<0,
\end{align*}
and
\begin{align*}
	{\rm det} \left[ \begin {array}{cc} \frac{\partial^2 f}{\partial x^2} 	&\frac{\partial^2 f}{\partial x\partial y} \\
	\noalign{\medskip}\frac{\partial^2 f}{\partial x \partial y}	&\frac{\partial^2 f}{\partial y^2}\end {array} \right]_{P} =-\frac{(\alpha-1)^4}{\alpha^2\beta^2}<0. 
\end{align*}
and hence,  $P$ is a generic transcritical point for the system \eqref{sf_model2}. The point $P$ is also a jump point, as at this point the fast flow \eqref{sf_fast subsystem} is moved away from the critical manifold $M_{10}$. 

We now consider a singular slow-fast cycle $\Gamma_0$ as follows: 
From $S(0, y_m)$, follow the slow flow \eqref{sf_slow subsystem} down the $y$-axis to $T(0,y')$, follow the fast flow \eqref{sf_fast subsystem} to intersect the attracting branch $S_0^{a}$ at $T'(x', y')$, follow the slow flow \eqref{sf_slow subsystem} along $S_0^{a}$ to $Q$ and then follow the fast flow \eqref{sf_fast subsystem} to the left of $Q$ to $S(0, y_m)$. Consequently, we have a singular orbit $\Gamma_0$ consisting of slow and fast segments for which $T$, $Q$ are jump points and $T', S$ are drop points as at these points, the fast flow is moved toward the critical manifolds.  

\begin{theorem}\label{sf_relaxation}
Let $\mu\in R_3$, $x_{2*}<x_m$ and $N$ be a tubular neighbourhood of $\Gamma_0$. Then for each fixed $0 < \epsilon \ll 1$, the system \eqref{sf_model2} has a unique relaxation oscillation $\Gamma_\epsilon\subset N$ which is strictly attracting with characteristic multiplier bounded by $-K/\epsilon$ for some constant $K > 0$. Moreover, the cycle $\Gamma_\epsilon$ converges to $\Gamma_0$ in the Hausdorff distance as $\epsilon \to 0$. 
\end{theorem}

\begin{proof}
 Conditions stated in the theorem ensure that the system \eqref{sf_model2} has a unique interior equilibrium $E_{2*}(x_{2*}, y_{2*})$ and the equilibrium lies to the left of the generic fold point $Q$. For $\epsilon>0$ very small, following the Fenichel's theorem $S_0^a$, $S_0^{a+}$, perturb to nearby slow manifolds $S_\epsilon^{a}$ and $S_\epsilon^{a+}$ and by theorem (2.1) of \cite{krupa2001extending}, the slow manifolds $S_\epsilon^a$ can be continued beyond the generic fold point $Q$ and by theorem (2.1) of  \cite{krupa2001extending1}, the slow manifold $S_\epsilon^{a+}$ can be continued beyond the generic transcritical point $P$. The slow manifold $S_\epsilon^a$ (resp. $S_\epsilon^{a+}$) lies close to $S_0^a$ (resp. $S_0^{a+}$) until it arrives at the vicinity of the generic fold point $Q$ (resp. generic transcritical point $P$).

We consider a small vertical section $\Delta=\{(x_0, y)|y\in[y_m-\epsilon_0, y_m+\epsilon_0]\}$, $0<\epsilon_0\ll 1$ 
We know that for every point $(0,y_0)$, $y_0\in [y_m-\epsilon_0, y_m+\epsilon_0]$ we can define $p_0(y_0)$ such that $0<p_0(y_0)<\frac{\beta}{\alpha-1}$ for $\gamma-\delta\geq 0$ or $\delta-\gamma<p_0(y_0)<\frac{\beta}{\alpha-1}$ for $\gamma-\delta<0$ by the result derived in lemma \ref{lemma_entry-exit}.
          
We now follow tracking two trajectories $\Gamma_\epsilon^{1,2}$ starting on $\Delta$ at the points $(x_0, y^{1,2})$. For $0 < \epsilon \ll 1$, it  follows by Fenichel's theorem that $\Gamma_\epsilon^{1,2}$ get attracted toward the slow manifold $S_\epsilon^{a+}$ exponentially with a rate $\mathcal{O}(e^{-1/\epsilon})$ and move downward slowly. Then by theorem (2.1) of \cite{krupa2001extending} $\Gamma_\epsilon^{1,2}$ pass by the generic transcritical point $P$ contracting exponentially toward each other and leave the repelling branch $S_0^{r+}$ of the critical manifold $M_{10}$ at the points $(0, p_0(y^{1,2}))$ and then jump horizontally to $(x_0, p_\epsilon(y^{1,2}))$ where $\displaystyle\lim_{\epsilon\to 0}p_\epsilon(y^{1,2})= p_0(y^{1,2})$. The trajectories then follow two layers of the fast flow \eqref{sf_fast subsystem} and get attracted towards the slow manifold  $S_\epsilon^a$ and pass the generic fold point $Q$ contracting exponentially and thus, finally return to $\Delta$.

Tracking the forward trajectories, we thus have a return map $\Pi: \Delta \to \Delta$ inducted by the flow of \eqref{sf_model2} for $0 < \epsilon \ll 1$. The return map $\Pi$ is a contraction map as the trajectories contract toward each other with rate $\mathcal{O}(e^{-1/\epsilon})$ and by the contraction mapping theorem $\Pi$ has a unique fixed point which is stable. This fixed point is the desired limit cycle $\Gamma_\epsilon$ which exists in a tubular neighbourhood of the singular slow-fast cycle $\Gamma_0$ and as the contraction is exponential, the characteristic multiplier of $\Gamma_\epsilon$ is bounded above by $-K/\epsilon$ for some $K>0$. Again applying Fenichel's theorem, theorem (2.1) of \cite{krupa2001extending1} and theorem (2.1) of \cite{krupa2001extending}, we conclude that the periodic orbit $\Gamma_\epsilon$ converges to the singular orbit $\Gamma_0$ as $\epsilon\to 0$ in the Hausdorff distance.
\end{proof}

For a geometrical description of the proof of theorem \ref{sf_relaxation}, 
see Fig. \ref{fig:Relaxation}.
\begin{figure}[H]
\setlength{\belowcaptionskip}{-10pt}
\centering
	\subfloat[]{\includegraphics[width=8cm,height=6cm]{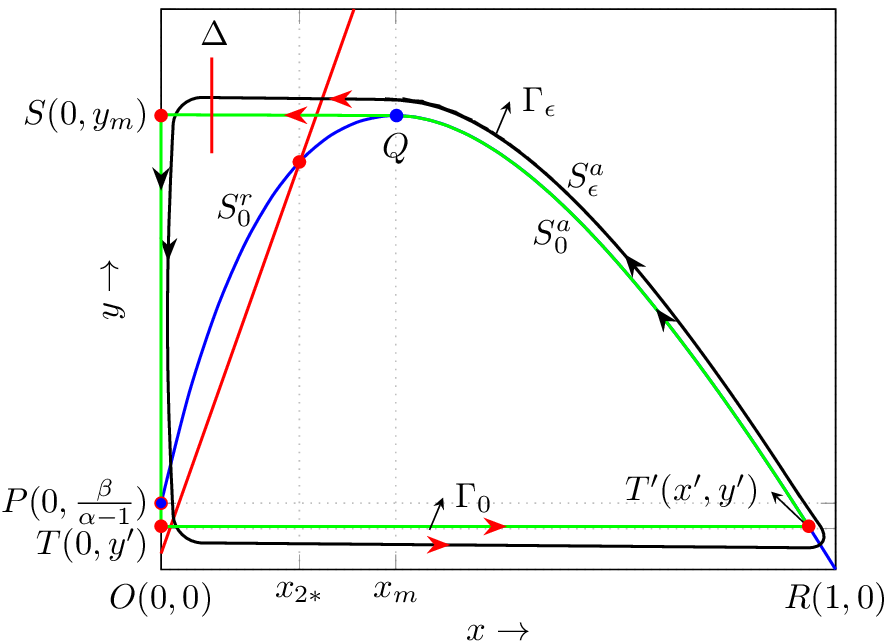}}
	\subfloat[]{\includegraphics[width=8cm,height=6.5cm]{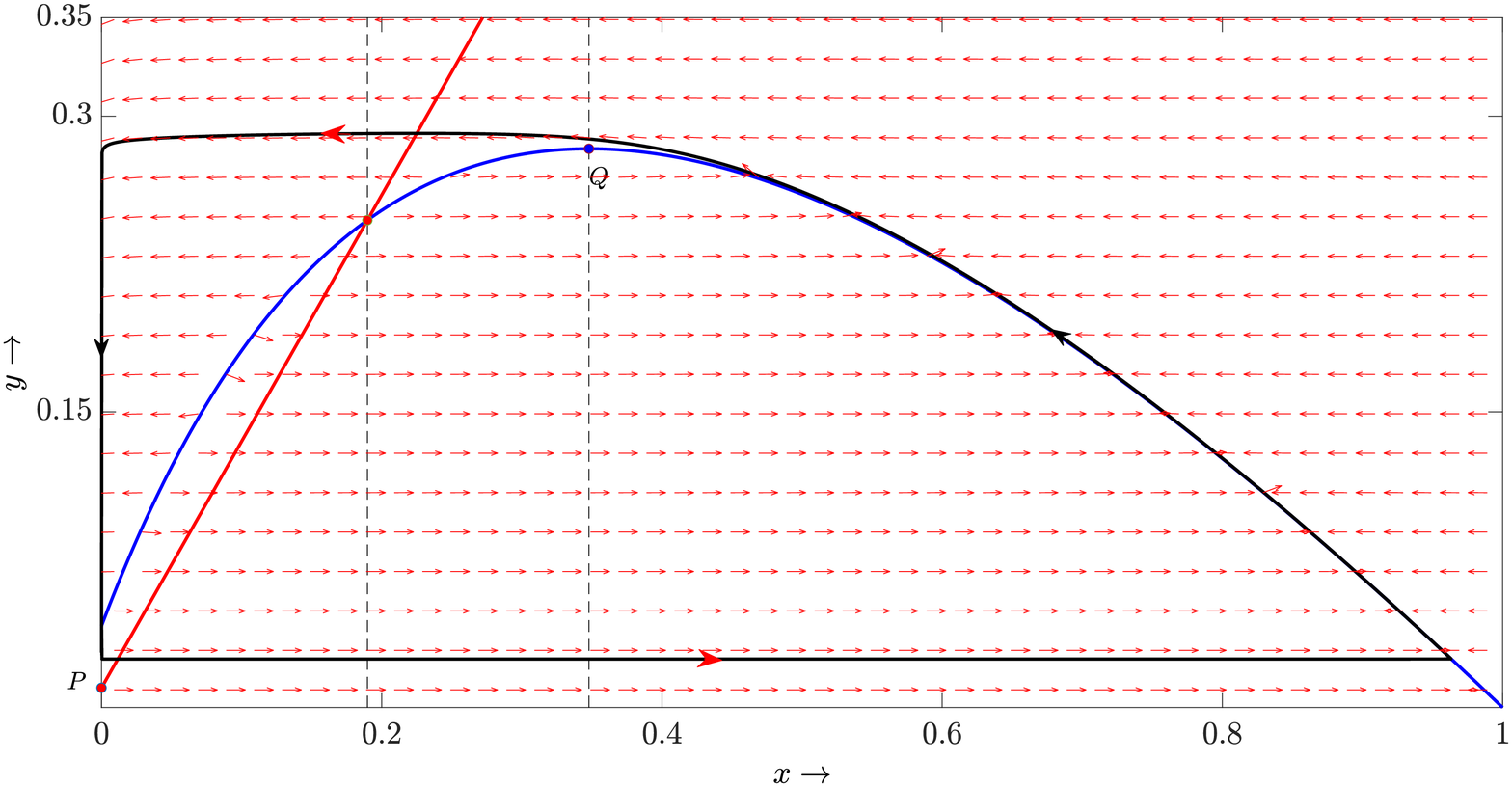}}
\caption{(a) Representation of the singular orbit $\Gamma_0$  (green) and a relaxation oscillation orbit $\Gamma_\epsilon$ (black) for the system \eqref{sf_model2} with $0<\epsilon\ll 1$ for $x_{2*}<x_m$.
The blue curve and red line, respectively, represent the non-trivial prey and predator nullclines. The two non-hyperbolic points $P$ and $Q$, shown by the solid blue circles, are the generic transcritical and generic fold points, respectively. The trajectory of $\Gamma_0$ is as follows: (i) it travels from $S(0,y_m)$ to $T(0,y')$ along the $y$-axis downward following the slow-flow \eqref{sf_slow subsystem}; (ii) it travels from $T(0,y')$ to $T'(x',y')$ parallel to $x$-axis following a layer of the fast subsystem \eqref{sf_fast subsystem}; (iii) it travels from $T'$ to $Q$ along the attracting branch $S_0^a$ following the slow-flow \eqref{sf_slow subsystem} and (iv) finally, travels from $Q$ to $S$ following the fast-flow \eqref{sf_fast subsystem}. According to the Fenichel's theorem, the submanifolds $S_0^{r}$ and $S_0^a$ could be perturbed to $S_{\epsilon}^{r}$ and $S_{\epsilon}^a$ for,  $0<\epsilon \ll  1$ respectively. These submanifolds are located within $\mathcal{O}(\epsilon)$ distance from $S_0^{r}$ and $S_0^a$, respectively. The vertical section $\Delta$ is defined in the text. (b) A numerical illustration of a relaxation oscillation (thick black periodic orbit) for the system \eqref{sf_model2} encompassing the unique interior equilibrium point obtained for the parameter values $\alpha=1.5, \beta=0.0207, \gamma=0.3,\delta=0.31,$ $\theta=1.25$ and $\epsilon=0.005$. Slow flow is represented by black arrows, whereas fast flow is shown by red arrows.  (For interpretation of the references to colour in this figure caption, the reader is referred to the web version of this chapter.)}
\label{fig:Relaxation}			
\end{figure}

\subsection{Bi-stability}
A phenomenon of bi-stability will occur for the system \eqref{sf_model2} 
whenever $\mu\in R_1$ and $x_{2*}>x_m$. Under these parametric restrictions, both the equilibria $E_2(0,\delta-\gamma)$ and $E_{2*}$ are stable, whereas the equilibrium $E_{1*}$ is a hyperbolic saddle. Thus, we have the basins of attraction for the equilibria $E_{2*}$ and $E_2$ separated by the stable and unstable separatrices of the saddle equilibrium $E_{1*}$. Geometrically, this signifies that if initially, the prey species lie to the left of the stable separatrix of $E_{1*}$, then at a long run the prey species will die out otherwise both the species will coexist.  

A common phenomenon for the model system \eqref{sf_model2} is that the generalist predator species have the choice of alternate food source when it's preferred prey is absent. This observation is reflected in the model system when $\mu\in R_4$, as in this case, there is no interior equilibrium for the system \eqref{sf_model2} but the axial equilibrium $E_2(0,\delta-\gamma)$ is a stable singularity. Thus, for $\mu\in R_4$, the prey species will die out in the long run, but the predator species will exist in spite of having Allee effect in the predator species. 
\begin{figure}[H]
\setlength{\belowcaptionskip}{-10pt}
\centering
\includegraphics[width=15cm,height=10cm]{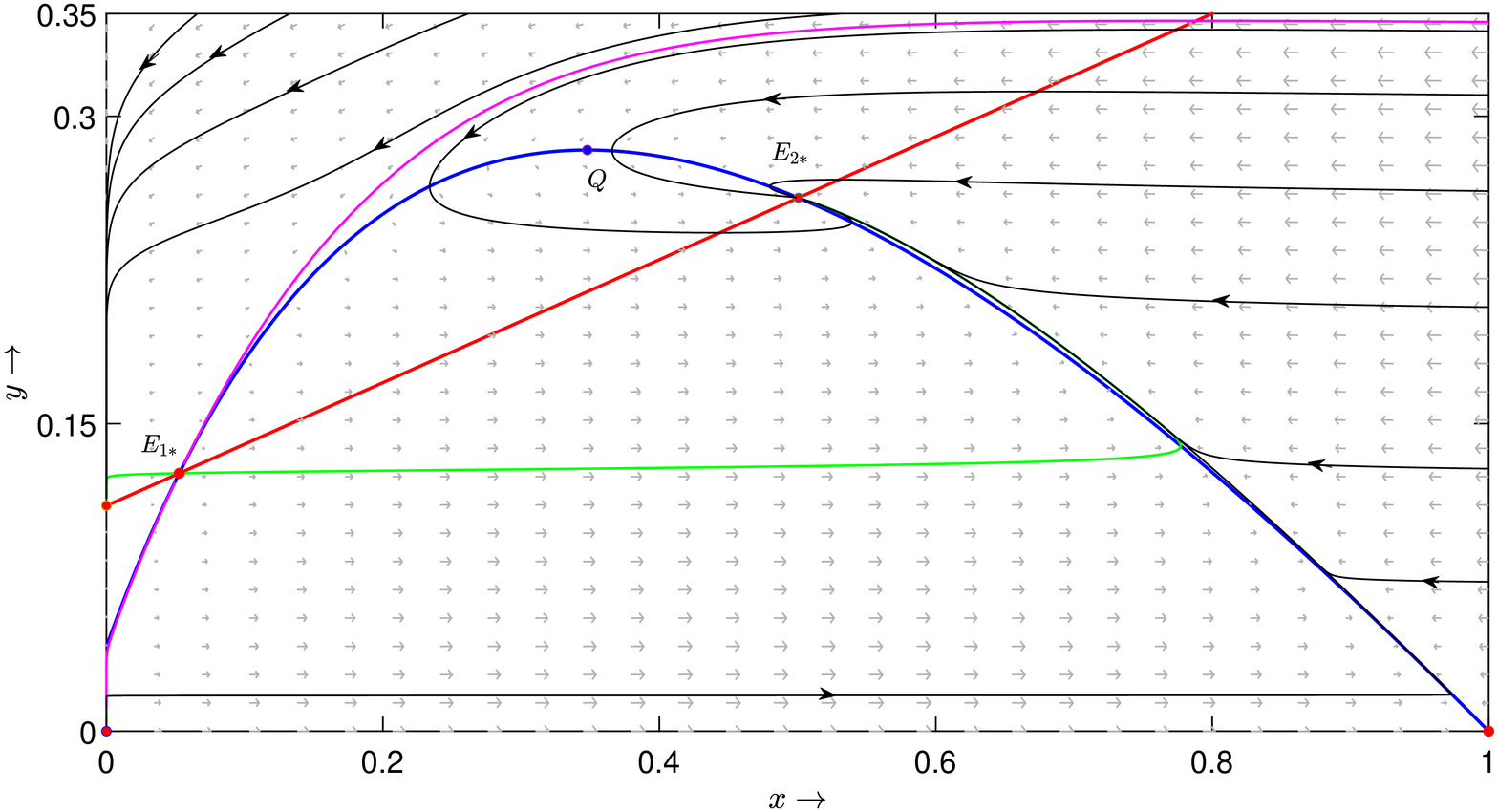}
\caption{The bistability scenario is shown in the phase portrait of the system \eqref{sf_model2} for $\mu\in R_1$ and $x_{2*}>x_m$. Non-trivial prey and predator nullclines are shown by the solid blue and red curves, respectively. The stable and unstable manifolds of the saddle equilibrium point $E_{1*}$ are shown by green and magenta curves, respectively. The phase portrait depicts that the system has two basins of attraction of $E_{2b}(0,\delta-\gamma)$ and $E_{2*}$ separated by the stable manifold (magenta) of the saddle equilibrium point $E_{1*}$.  The parameter values used are given by $\alpha=1.5, \beta=0.0207, \gamma=0.3, \delta=0.41$, $\theta=0.3$ and $\epsilon=0.1$.  (For the interpretation of the colour references in this figure caption, the reader is referred to the web version of this chapter.)}
\label{fig:bistability}			
\end{figure}

\section{Conclusion}\label{Sec:conclusion}
In this work, we focus on the investigation of the dynamics of a planar slow-fast modified Leslie-Gower predator-prey model with a weak Allee effect in the predator based on the natural assumption that the prey reproduces far more quickly than the predator. The Allee effect and its significance in population growth and decline have been extensively discussed in the empirical and scientific literature. In this article, the Allee effect is incorporated in the predator population due to the fact that the Allee effect has a significant impact on dynamics, especially boosting population decline and population extinction \cite{boukal2007predator}. Predators may experience the Allee effect for a variety of reasons, including poor sperm quality, a paucity of suitable mates, poor fertilization rate, or even cooperative breeding \cite{courchamp2008allee,berec2007multiple}. Here, we are assuming that the Allee effect exists exclusively in predators whose predation behaviour is governed by the Beddington-DeAngelis functional response. Furthermore, the predator population is thought to be a generalist predator, meaning that if  the preferred prey is not readily accessible, the predator population might switch to consuming some additional food. 

We have applied  Fenichel’s theorems for normally hyperbolic critical manifold, and the blow-up method is used to fully understand the geometry of the manifolds and how they cross at non-hyperbolic places. The parameter space has been divided into four regions $R_1, R_2, R_3$ and $R_4$ and the dynamical analysis of the system is performed in the various parametric regions.

It has been observed that the canard point $Q$, which exists on the parabolic critical curve $M_{20}$, serves as an organizing centre for complicated dynamics, as in a neighbourhood of the canard point, we have detected various rich phenomena including canard cycles due to singular Hopf bifurcation, the birth of canard explosions (transition from a small amplitude canard cycle to a large amplitude relaxation oscillation for $\mu\in R_3$), relaxation oscillation, homoclinic orbits and heteroclinic orbits. The occurrence of relaxation oscillations is demonstrated by the entry-exit function. In order to demonstrate the existence of homoclinic orbits, homoclinic to the hyperbolic saddle $E_{1*}$ and enclosing the canard point; heteroclinic orbits (unique or infinitely many) connecting various equilibria, we present a thorough mathematical study. This exhibits the long-term behaviour of the system. Each of the above results is verified numerically for the choice of parameter values in various regions.

The presence of relaxation oscillation and the inception of the canard explosion are particularly relevant to this discussion, as they have important ecological implications. The thorough mathematical results in this regard are also presented in previous sections. The large amplitude relaxation oscillation as $0<\epsilon\ll 1$ exists in the region $R_3$ has a shape with two segments parallel to $x$-axis, one  segment parallel to $y$-axis, and a curved segment similar to the attracting branch $S_0^a$ of the slow manifold $M_{20}$. From a biological perspective, the presence of the relaxation oscillation is an indication of the possibility of prey and predator living together. A prey breakout happens in a very short amount of time if the density of predators drops to a level that is lower than the lowest value on the critical curve. Once the density of the prey reaches a level that is enough to sustain the reproduction of the predators, the density of the predators will continue to gradually increase over an extended period of time until it reaches a density that is higher than the maximum value of the critical curve. The number of potential prey is thus decreasing noticeably over a shorter period of time. When there is less food available (prey), the number of predators steadily drops over a shorter period of time. After some time has passed, the reduced prey population density leads the number of predators to fall as well. This ensures that the cycle will continue, allowing populations of both prey and predator to coexist. The canard explosion is a very surprising and interesting phenomenon, which occurs in an exponentially small range of the parameter $\delta$. The change from a small amplitude canard cycle to a large amplitude relaxation oscillation takes place as a result of a sequence of canard cycles that occur very near to the canard point, follow the repelling slow manifold, and then jump to one of the attracting slow manifolds. From a biological point of view, this canard explosion might be seen as an early warning sign of a forthcoming regime shift as a result of an exponentially small change of parameter $\delta$.

The important part of the conclusion is that we have obtained various co-dimension 1 and 2 bifurcation structures of the slow-fast model, including the saddle-node, Hopf, transcritical bifurcation, generalized Hopf, cusp point, homoclinic, heteroclinic and Bogdanov-Takens bifurcations and used bifurcation diagrams to support the outcomes of the bifurcation. The boundary curves that divide the parametric space into different domains in the two-parameter bifurcation diagram are called bifurcation curves (see Fig. \ref{fig:regions_bifurcation}). A qualitative change in the system has been observed and reflected through the bifurcation analysis as we shift from one region to the other through the parameter variations. The following scenarios show that bi-stability exists when the parameters are in the region $R_1$:
\begin{enumerate}
\item The co-existing equilibrium $E_{2*}$ lies right to the fold point $Q$. In such a case, the basins of attraction for the co-existing equilibrium $E_{2*}$ and the prey-free equilibrium $E_{2b}$ are separated by the stable and unstable separatrices of the hyperbolic saddle $E_{1*}$. Thus, we have the stable coexistence of both the interior equilibrium and also the prey-free equilibrium in spite of having Allee effect in the predator.
\item We also have coexistence of stable oscillation around the canard point and the prey free equilibrium point $E_{2b}$, and these phenomena have been realized in a neighbourhood of the GH point when the parameters belong to the region $R_1$ bounded by the Hopf curve and LPC curve.
\end{enumerate}
From a biological point of view, such bi-stability indicates that the system exhibits either \lq prey extinction', \lq stable coexistence of the populations', or \lq oscillating coexistence of the populations as a result of the appearance of the LPC curve emerging from the GH point' depending on the initial population size and values of system parameters. Moreover, the global stability of the unique positive equilibrium point $E_*$ for $\mu\in R_3$, $x_*>x_m$, $0<\epsilon\ll 1$ is investigated. 

At the end of this study, we would like to bring attention to the fact that a further interesting and challenging work of determining cyclicity by using slow divergence integral of various limit periodic sets like the canard slow-fast cycles with or without head, canard point when $A=0$, etc will be performed in near future.

\section*{Acknowledgement}
 The first author acknowledges the financial support from the Department of Science and Technology (DST), Govt. of India, under the scheme “Fund for Improvement of S\&T Infrastructure (FIST)” [File No. SR/FST/MS-I/2019/41].

\section*{Conflict of interes}
The authors declare that they have no known competing financial interests or personal relationships that could have appeared to influence the work reported in this paper.
\section*{Data availability Statement:}
Data sharing is not applicable to this article, as no datasets were generated or analysed during the current study.

\nocite{*} 
{1}
\end{document}